\documentclass[a4paper, 11pt]{article}

\usepackage[english]{babel}
\usepackage[utf8x]{inputenc}
\usepackage[T1]{fontenc}

\usepackage[a4paper,top=2.5cm,bottom=2.5cm,left=2.5cm,right=2.5cm,marginparwidth=1.75cm]{geometry}

\usepackage{amsmath,amssymb, bm}
\usepackage{graphicx}
\usepackage[colorinlistoftodos]{todonotes}
\usepackage[colorlinks=true, allcolors=blue]{hyperref}
\usepackage{amsthm}
\usepackage{algorithm}
\usepackage[noend]{algpseudocode}
\usepackage{multirow}
\usepackage{harvard}

\usepackage{mathtools}

\newtheoremstyle{break}
  {\topsep}{\topsep}%
  {\itshape}{}%
  {\bfseries}{}%
  {\newline}{}%
\theoremstyle{break}

\newtheorem{example}{Example}
\newtheorem{lemma}{Lemma}
\newtheorem{remark}{Remark}
\newtheorem{definition}{Definition}
\newtheorem{assumption}{Assumption}

\newcommand{\ci}{\citeasnoun}
\newcommand{\fil}{\mathcal{F}}
\newcommand{\x}{\mathcal{X}}
\newcommand{\bbr}{\mathbb{R}}
\newcommand{\p}{\mbox{\sc Price}}

\title{Simulation Methods for Robust Risk Assessment\\ and the Distorted Mix Approach}
\author{Sojung Kim \hspace{1cm} Stefan Weber   \\[1.0ex] \textit{Leibniz Universit{\"a}t Hannover} }
\date{May, 2021 \footnote{Institute of Actuarial and Financial Mathematics \& House of Insurance, Leibniz Universit\"at Hannover, Welfengarten 1, 30167 Hannover, Germany.
e-mail: {\tt sojung.kim@insurance.uni-hannover.de}, {\tt stefan.weber@insurance.uni-hannover.de}.}}

\begin{document}
\maketitle

\begin{abstract}
Uncertainty requires suitable techniques for risk assessment. Combining stochastic approximation and stochastic average approximation, we propose an efficient algorithm to compute the worst case average value at risk in the face of tail uncertainty. Dependence is modelled by the distorted mix method that flexibly assigns different copulas to different regions of multivariate distributions. We illustrate the application of our approach in the context of financial markets and cyber risk. 
\end{abstract}

\vspace{5mm}


\vspace{5mm}

\section{Introduction}\label{sec1:intro}

Capital requirements are an instrument to limit the downside risk of financial companies. They constitute an important part of banking and insurance regulation, for example, in the context of Basel III, Solvency II, and the Swiss Solvency test. Their purpose is to provide a buffer to protect policy holders, customers, and creditors. Within complex financial networks, capital requirements also mitigate systemic risks. 

The quantitative assessment of the downside risk of financial portfolios is a fundamental, but arduous task. The difficulty of estimating downside risk stems from the fact that extreme events are rare; in addition, portfolio downside risk is largely governed by the tail dependence of positions which can hardly be estimated from data and is typically unknown. Tail dependence is a major source of model uncertainty when assessing the downside risk.

In practice, when extracting information from data, various statistical tools are applied for fitting both the marginals and the copulas -- either (semi-)parametrically or empirically. The selection of a copula is frequently made upon mathematical convenience; typical examples include Archimedean copulas, meta-elliptical copulas, extreme value copulas, or the empirical copula, see e.g.~\ci{QRM}. The statistical analysis and verification is based on the available data and is center-focused due to limited observations from tail events. This approach is necessarily associated with substantial uncertainty.  The induced model risk thus affects the computation of monetary risk measures, the mathematical basis of capital requirements. These functionals are highly sensitive to tail events by their nature -- leading to substantial misspecification errors of unknown size.

In this paper, we suggest a novel approach to deal with this problem. We focus on the downside risk of portfolios. Realistically, we assume that the marginal distributions of individual positions and their copula in the central area can be estimated sufficiently well. We suppose, however, that a satisfactory estimation of the dependence structure in the tail area is infeasible. Instead, we assume that practitioners who deal with the estimation problem share viewpoints on a collection of copulas that potentially capture extremal dependence. However, practitioners are uncertain about the appropriate choice among the available candidates.

The family of copulas that describes tail dependence translates into a family of joint distributions of all positions and thus a collection of portfolio distributions. To combine the ingredients to joint distributions, we take a particularly elegant approach: 
The Distorted Mix (DM) method developed by \ci{LYY14} constructs a family of joint distributions from the marginal distributions, the copula in the central area and several candidate tail copulas. A DM copula is capable of handling the dependence in the center and in the tail separately. We use the DM method as the starting point for a construction of a convex family of copulas and a corresponding set of joint distributions.\footnote{ Simple mixture approaches were previously considered in the literature, but are less flexible than the DM approach. An early contribution with an application to financial markets is \ci{hu2006}.}

Once a family of joint distributions of the positions is given, downside risk in the face of uncertainty can be computed employing a classical worst case approach. To quantify downside risk, we focus on robust average value at risk (AV@R). The risk measure AV@R is the basis for the computation of capital requirements in both the Swiss solvency test and Basel III. As revealed by the axiomatic theory of risk measures, AV@R has many desirable properties such as coherence and sensitivity to tail events, see \ci{SF}. In addition, AV@R is m-concave on the level of distributions, see  \ci{BB15}, and admits the application of well-known optimization techniques as described in \ci{RU00} and \ci{RU02}. 

To be more specific, we consider a $d$-dimensional random vector  ${\bm X} = (X_1, X_2, \dots, X_d)$ with given marginals and a copula $C$ in a set of distorted mix copulas $\mathcal{DM}_{\tilde{\mathfrak{C}}}$. Considering aggregate losses $X \; = \; \Psi(X_1, \cdots, X_d)$ for some measurable function $\Psi$, we study the worst case risk 
$$
        \max_{C \in \mathcal{DM}_{\tilde{\mathfrak{C}}} } \rho(X) 
$$
where $\rho$ signifies AV@R at some fixed level. The exact problem will be described in Section \ref{sec:wcra}.

Our model setup leads to a continuous stochastic optimization problem to which we apply a combination of stochastic approximation and sample average approximation. We explain how these techniques may be used to reduce the dimension of the mixture space of copulas. We discuss the solution technique in detail and illustrate its applicability in several examples. 
\par \vspace{3mm}
\noindent The main contributions of the paper are:
\begin{enumerate}
\item For a given family of copulas modeling tail dependence, we describe a DM framework that conveniently allows worst case risk assessment.
\item We provide an efficient algorithm that numerically computes the worst case risk and identifies worst case copulas in a lower-dimensional mixture space.
\item We successfully apply our framework in selected case studies. The considered examples are financial markets and cyber risk.
\end{enumerate}

The paper is structured as follows. Section \ref{sec2:problem} explains the DM approach to model uncertainty and formulates the optimization problem associated to the computation of robust AV@R. In Section \ref{sec3:sto_opt}, we develop an optimization solver combining stochastic approximation (i.e.,~the projected stochastic gradient method) and sample average approximation: stochastic approximation identifies candidate copulas and a good approximation of the worst-case risk; in many cases, risk is insensitive to certain directions in the mixture space of copulas, enabling us to use sample average approximation to identify worst-case solutions in lower dimensions.  Section \ref{sec4:app} discusses two applications of our framework, namely to financial markets and cyber risk. Section \ref{sec5:conc} concludes with a discussion of potential future research directions.

\subsection*{Literature}

The concept of model uncertainty or robustness is a challenging topic in practice that has also been intensively discussed in the academic literature. Risk refers to situations in which possible scenarios and their associated probabilities are known; uncertainty in contrast describes circumstances in which the probabilities of events are not known, but are characterized by a nontrivial set of probability measures. This paper considers model uncertainty of this type. 

The underlying key assumption concerns the structure of the considered probability measures. We assume that marginal distributions and the dependence structure of typical events is known; uncertainty arises from tail dependence. Risk is measured on the basis of a worst-case approach, a suitable algorithm is suggested and implemented in case studies. These structural assumptions on the class of probability measures are motivated by typical considerations in practice and distinguish our analysis from previous approaches in the literature. 

Our methodology parallels the one chosen by \ci{GX14} in the sense that both papers study the worst-case in a given class of models. However, \ci{GX14} consider probability measures that are in some neighborhood of a given benchmark model in terms of relative entropy or another divergence measure. Similarly, \ci{HH13} and \ci{BC16} study ambiguity sets defined by relative entropy. While divergence measures are not proper metrics, optimal transport costs such as Wasserstein distances are metrics under regularity conditions on the cost function;  \ci{BM19} investigate model uncertainty in such a framework. \ci{BDT20} use a similar approach and study robust optimized certainty equivalent risk measures in the context of optimal transport costs. 

Another, complementary perspective on robustness comes from statistics, as suggested by \ci{hampel1971}. By Hampel's famous theorem, the classical notion of robustness can be characterized by continuity properties of functionals with respect to the weak topology. Distribution-based convex risk measures such as average value of risk are not continuous in this sense and thus not Hampel-robust, see  \ci{CDG10} and \ci{kou}. A refined notion of Hampel robustness that corresponds to finer topologies on sets of probability measures is suggested in the seminal paper \ci{KSZ14} and applied to risk measures; the proposal in  \ci{KSZ14} lifts the corresponding issues of robust statistics to a higher level that permits a comprehensive analysis.

In the current paper, we focus on worst-case AV@R in a multi-factor model. We are interested in the worst-case risk if the dependence structure is uncertain. This is closely related to papers that derive bounds in the face of partial information about dependence, cf. \ci{EPR13}, \ci{BJW14}, \ci{BV15},  \ci{Ru16}, \ci{PRSV17}, \ci{LSWY18}, \ci{EHR18}, \ci{WEBER2018191}, and \ci{HKW20}. In contrast to these contribution, we propose an algorithmic DM approach that is based on candidate copulas; this setting is very flexible in terms of the marginal distributions and copulas that are considered. Our framework is suitable for capturing uncertainty about the dependence in specific regions of a joint distribution. A simpler setting of mixture distributions is studied in \ci{ZF09} and \ci{KR14}.

Our algorithm builds on sampling-based stochastic optimization techniques. Applications of stochastic approximation and stochastic average approximation to the evaluation of risk measures were investigated by \ci{RU00}, \ci{RU02}, \ci{DuWe07}, \ci{BFP09},  \ci{DW10}, \ci{MSG10}, \ci{SXW14},  and \ci{BFP16}. The contributions discuss different risk measures including AV@R and utility-based shortfall risk, efficient estimation including variance reduction, portfolio optimization and hedging but do not concentrate on model uncertainty.  Without a specific focus on risk measures, the relevant simulation techniques are also discussed in  \ci{KY03},  \ci{Shapiro03},  \ci{Fu06}, \ci{BPP13}, and \ci{KPH15}.  \ci{GL19} study worst-case approximations for performance measures in the face of uncertainty, based on stochastic approximation; their analysis does, however, not specifically consider uncertainty about dependence in different regions of multivariate distributions as captured by the DM approach in this paper.

\section{The Distorted Mix Approach to Model Uncertainty} \label{sec2:problem}

\subsection{Distorted Mix Copula}

Letting $(\Omega, \fil, P)$ be an atomless probability space, we consider the family of random variables $\x= L^1(\Omega, \fil, P)$.  The task consists in computing the risk $\rho (X)$ of an aggregate loss random variable $X\in \x$ for a risk measure $\rho$. A finite distribution-based monetary risk measure $\rho: \x \to \bbr$ is a functional with the following three properties:
\begin{itemize}
\item {\bf Monotonicity}: $X \leq Y \; \Rightarrow \; \rho(X) \leq \rho(Y) \quad \forall X, Y \in \x$
\item {\bf Cash-invariance}: $\rho(X+  m ) = \rho(X) + m \quad \forall X\in \x, m \in \bbr$
\item {\bf Distribution-invariance}: $P\circ X^{-1} = P \circ Y^{-1} \; \Rightarrow \; \rho(X) = \rho (Y) \quad  \forall X, Y \in \x$
\end{itemize}

We consider a specific factor structure of aggregate losses. We assume that 
$$X \; = \; \Psi(X_1, \cdots, X_d) \; \in \; \x$$ 
where ${\bm X} = (X_1, \cdots, X_d)$ is a $d$-dimensional random vector and $\Psi: \mathbb{R}^d \rightarrow \mathbb{R}$ is some measurable function. The individual components $X_i$ may depict different business lines, risk factors, or sub-portfolios, and the function $\Psi: \mathbb{R}^d \rightarrow \mathbb{R}$ summarizes the quantity of interests. Frequently used aggregations are the total loss $X = \sum_{i=1}^d X_i$ and the excess of loss treaty $X = \sum_{i=1}^d (X_i-k_i)^+$ for thresholds $k_i \in \mathbb{R}^+$.  

Computing the risk measure $\rho(X)$ requires a complete model of the random vector ${\bm X} = (X_1, \cdots, X_d)$. Let $F(x_1, \cdots, x_d)$ be its unknown d-dimensional joint distribution which we aim to understand. By Sklar's theorem, any multivariate distribution $F$ can be written as the composition of a copula $C$ and the marginal distributions $F_i$ of its components:
\begin{equation*} \label{eq:sklar}
F(x_1, \cdots, x_d) = C(F_1(x_1), \cdots, F_d(x_d)).
\end{equation*}
The typical situation in practice is as follows: 
\begin{itemize}
\item The marginals $F_1(x_1), \cdots, F_d(x_d)$ and the dependence structure in the central area, denoted by the copula $C_0$, can be estimated from available data. 
Typical examples of $C_0$ may include the Gaussian copula, the t-copula, or the empirical copula.  
\item However, due to limited observations in the tail, the copula $C_0$ might not capture the characteristics of the extreme area very well. Instead, in the face of tail uncertainty, extreme dependence should be captured by a collection of copulas instead of a single copula.  This will be explained in Section~\ref{sec:wcra}.
\end{itemize}

Before we describe our approach to model uncertainty in the next section, we introduce an important tool for combining different copulas in order to to handle the central and tail parts separately, the \emph{Distorted Mix (DM) method}, see \ci{LYY14}.  A DM copula $C$ is constructed from $m+1$ component copulas: $C_0$ for the typical area, and $C_1, \cdots, C_m$ for the extreme area.

\begin{definition}[Distorted mix copula]\label{def1}
Let $D_{ij}:[0,1] \rightarrow [0,1]$ be continuous distortion functions, i.e.,~continuous, increasing functions with $D_{ij}(0) = 0, D_{ij}(1) = 1$, and $\alpha_i \geq 0$, $i=0, \cdots, m$, $j = 1, \cdots, d$, such that
    \begin{equation}\label{ass:dmc}
        \sum_{i=0}^m \alpha_i \; = \; 1, \quad\quad\quad\quad \sum_{i=0}^m \alpha_i D_{ij}(v) \; = \; v \quad \forall \, v \in [0,1], \, j= 1, \cdots, d.
    \end{equation}
    For any collection of copulas $C_0, \cdots, C_m$, the corresponding \emph{distorted mix copula} $C: [0,1]^d \rightarrow [0,1]$ is defined by
    \begin{equation}\label{def:dmc}
        C(u_1, \cdots, u_d) \; = \; \sum_{i=0}^m \alpha_i C_i(D_{i1}(u_1), \cdots, D_{id}(u_d)).
    \end{equation}
     \end{definition}

\begin{remark}
A copula captures the dependence structure of a multivariate random vector with marginal distribution functions $F_1, F_2, \dots, F_d$ as a function of $u_1 = F_1 (x_1), u_2 = F_2 (x_2) \dots, u_d = F_d (x_d)$ with $x_1, x_2, \dots, x_d \in \bbr$. The argument $x_j$ is a quantile of $F_j$ at level $u_j$, $j=1,2, \dots, d$: Levels close to $0$ correspond to the lower tail of $(X_1, X_2, \dots, X_d)$, levels close to 1 to the upper tail, and other levels to the center of the distribution.  

In equation \eqref{def:dmc}, for $i=0, \dots, m$, the parameter $\alpha_i$ defines the probability fraction of the total dependence that is governed by copula $C_i$ which is distorted by the distortion functions $D_{i1}, D_{i2}, \dots, D_{id}$. These distortion functions  describe how the arguments (or levels) of copula $C$ are mapped to the arguments (or levels) of the ingredient copulas $C_i$. We illustrate these features in the following example.
\end{remark}

\begin{example}
Let $d=m=2$ and $\alpha_0=\alpha_1=\alpha_2 = 1/3$. We suppose that $C_1$ and $C_2$ are the comonotonic copulas, i.e., $C_1(u_1, u_2) = C_2(u_1, u_2) = \min (u_1, u_2)$, and that $C_0$ is the countermonotonic copula, i.e., $C_0(u_1, u_2) = \max(u_1 + u_2 - 1, 0)$. We set $D_{ij}(u_j) = \max\{3 \cdot(u_j - a_i),0\} \wedge 1$, $a_1 = 0$, $a_2=2/3$, $a_0=1/3$, $j= 1,2$. Obviously, the lower and upper tails are governed by the comonotonic copulas $C_1$ and $C_2$, respectively, and the central part is countermonotonic according to $C_0$. In this particular example, the dependence structure in each part is exclusively controlled by one of the copulas $C_0$, $C_1$, and $C_2$.\footnote{ The example focuses on piecewise linear distortion functions with disjoint support. This is a special case of the distorted mix model in equation \ref{def:dmc} in which different parts of the dependence structure are each controlled by a single copula $C_1, \dots, C_m$. Other choices of distortion functions capture more general dependence structures.} 
\end{example}

\subsection{Worst-Case Risk Assessment}\label{sec:wcra}

In this section, we explain our approach to risk assessment in the face of tail uncertainty. As described in the previous section, we assume that the marginals of the random vector ${\bm X} = (X_1, X_2, \dots, X_d)$ are given. Its copula is unknown, but possesses the following DM structure:
\begin{itemize}
\item Let $\mathfrak{D} = \{ D_{ik}: \, i=0, \cdots, m, \, k = 1, \cdots, d \}$  be a collection of distortion functions and
$\bm \alpha = (\alpha_0, \cdots, \alpha_m) \in [0,1]^{m+1}$  satisfying assumption \eqref{ass:dmc}. 
\item In addition, we fix a copula $C_0$ and a set $\tilde {\mathfrak{C}} $ of copulas. 
\end{itemize}
We assume that the copula of ${\bm X} = (X_1, X_2, \dots, X_d)$ belongs to the following family:
 \begin{equation*}\label{def:D}
        \mathcal{DM}_{\tilde{\mathfrak{C}}} = \left\{ \alpha_0 C_0(D_{01}(u_1), \cdots, D_{0d}(u_d)) + \sum_{i=1}^m \alpha_i \tilde{C}_i(D_{i1}(u_1), \cdots, D_{id}(u_d)), \tilde{C}_i \in \tilde {\mathfrak{C}} \ \forall i=1, \cdots, m \right\}
    \end{equation*}
The worst-case risk assessment over all feasible distributions of ${\bm X} = (X_1, X_2, \dots, X_d)$ is equal to 
    \begin{equation}\label{eq:opt1}
        \max_{C \in \mathcal{DM}_{\tilde{\mathfrak{C}}} } \rho(X) 
    \end{equation}
where $X=\Psi(X_1, \cdots, X_d)$ and ${\bm X} = (X_1, \cdots, X_d)$ has a copula $C \in \mathcal{DM}_{\tilde{\mathfrak{C}}}$ with the given marginals. 

\begin{remark}
Our approach assumes that the marginals of ${\bm X}$ and the copula $C_0$ are known; if the distortion functions are suitably chosen, $C_0$ could, for example, govern the central area. However, dependence in other regions, e.g. tail dependence, is uncertain and captured by some family $\tilde{\mathfrak{C}}$ of copulas. The key structural assumption is that ${\bm X}$ possesses a DM copula and that the distortions and associated probability fractions are fixed. These determine the composition of the copula of ${\bm X}$. The distortions and probability fraction associated to the copula $C_0$ cannot be varied; for all other distortions and associated probability fractions the corresponding copulas may flexibly be chosen from the collection $\tilde{\mathfrak{C}}$.
\end{remark}

\begin{remark}
One possible approach would be to choose $\tilde{\mathfrak{C}}$ as a finite collection of $K\geq m$ candidate copulas. In this case, the number of the DM copulas is either  $\binom{K}{m} \times m!$ or $K^m$ if we allow duplicate components. This approach has two disadvantages: 
First, from a technical point of view the corresponding discrete optimization problem involves a very high number of permutations. Computing the value function for each of them is expensive, and the Ranking and Selection (R$\&$S) method would  not be efficient in this case.  Second, with finitely many candidate copulas  also their mixtures seem to be plausible ingredients to the DM method and should not be excluded a priori.  
\end{remark}

For a given collection $\mathfrak{C}=\{C_1, C_2, \dots, C_K\}$ of $K$ candidate copulas, we consider the family $\tilde{\mathfrak{C}}$ of their mixtures. That is, any element of $\tilde{\mathfrak{C}}$ can be expressed as a convex combination of elements of $\mathfrak{C}$:
    \begin{equation*}
        \tilde{C}^{\bm{\gamma}} \; = \; \sum_{j = 1}^K \gamma_j C_j, \quad \quad
       \bm{\gamma} \in  \bigtriangleup^{K-1} = \left\{ \bm{\gamma} = \begin{pmatrix}\gamma_1 \\ \gamma_2 \\ \vdots \\ \gamma_K\end{pmatrix} \in \mathbb{R}^{K} \left| \ \sum_{j=1}^K \gamma_j = 1 \mbox{ and } \gamma_j \geq 0 \mbox{ for all } j \right.\right\},
    \end{equation*}
    where $\bigtriangleup^{K-1}$ is the standard $K-1$ simplex. The $K$ vertices of the simplex are the points $\bm{e}_i \in \mathbb{R}^{K}$, where $\bm{e}_1 = (1,0, \cdots, 0)^\top , \bm{e}_2 = (0,1, \cdots, 0)^\top$, $\dots, \bm{e}_K = (0, 0, \cdots, 1)^\top $.
With this notation, our $K$ candidate copulas ${C}_j \in \mathfrak{C}$ can be written as $C_j = \tilde{C}^{\bm{e}_j}$.

Any element in $\mathcal{DM}_{\tilde{\mathfrak{C}}} $ can now be represented by some $\bar{ \bm \gamma} = (\bm{\gamma}^1, \cdots, \bm{\gamma}^m) \in \mathbb{R}^{K \times m}$ with $\bm{\gamma}^1, \cdots, \bm{\gamma}^m \in \bigtriangleup^{K-1}$ according to the following formula:
    \begin{equation}\label{def:newdmc}
        \hat{C}^{\bm{\gamma}^1, \cdots, \bm{\gamma}^m}(u_1, \cdots, u_d)  = \alpha_0 C_0(D_{01}(u_1), \cdots, D_{0d}(u_d)) + \sum_{i=1}^m \alpha_i \tilde{C}^{\bm{\gamma}^i}(D_{i1}(u_1), \cdots, D_{id}(u_d)).
    \end{equation}
With this notation, the optimization problem \eqref{eq:opt1} can be rewritten as 
    \begin{eqnarray}\label{eq:opt2}
        && \max_{\bm{\bar{\gamma}} = (\bm{\gamma}^1, \cdots, \bm{\gamma}^m) \in (\bigtriangleup^{K-1})^m  } \; \rho\left(X^{\bar{ \bm \gamma}} \right) 
    \end{eqnarray}
where $X^{\bar{ \bm \gamma}}$ represents the aggregate loss $\Psi(X_1, \cdots, X_d)$ with $(X_1, \cdots, X_d)$ having copula $\hat{C}^{\bm{\gamma}^1, \cdots, \bm{\gamma}^m}$ and the given marginals. We call $\hat{C}^{\bm{\gamma}^1, \cdots, \bm{\gamma}^m}$ in \eqref{def:newdmc} a \textit{robust DM copula} if it attains the optimal solution of \eqref{eq:opt2}. Optimization problem \eqref{eq:opt2} enables us to search the solution inside of the multiple simplexes  and paves a way to utilize the gradient approach. 

We will now construct and explore a sampling-based optimization solver. For this purpose, we focus on one particular risk measure, the average value at risk (AV@R), also called conditional value at risk or expected shortfall. This risk measure forms the basis of Basel III and the Swiss Solvency test. If $p \in (0,1)$ is the level of the AV@R, a number close to 1, the corresponding AV@R of  the losses $X^{\bar{\bm \gamma}}$ is defined as
    \begin{equation*}\label{def:CVaR}
        c_p(\bar{\bm{\gamma}}) = \frac{1}{1-p} \int_p^1 q_t^-\left( F_{X^{\bar{\bm \gamma}} } \right) \ dt
    \end{equation*}
where $q_t^-(F) = \inf \{x \in \mathbb{R} \ | \ F(x) \ge t \}$ for a distribution $F$ and $F_{X^{\bar{\bm \gamma}} }$ stands for the distribution of $X^{\bar{\bm \gamma}}$.
Accordingly, we denote VaR by $v_p(\bar{\bm{\gamma}}) = q_p^-\left( F_{X^{\bar{\bm \gamma}} } \right)$.
With this notation, our optimization problem is 
\begin{eqnarray}\label{eq:opt3}
        &
        & 
        \max_{\bm{\gamma}^1, \cdots, \bm{\gamma}^m \in \bigtriangleup^{K-1}}c_p\left( \bm{\gamma}^1, \cdots, \bm{\gamma}^m \right)  \nonumber \\
      & = &   \max_{\bm{\bar{\gamma}}  \in (\bigtriangleup^{K-1})^m  } c_p\left( \bm{\bar{\gamma}} \right). 
\end{eqnarray}

\subsection{Sampling Algorithm}\label{subs:Sampl}

\subsubsection{Portfolio Vector} 

The factor structure of DM copulas provides the basis for adequate simulation methods (see Proposition 1, \ci{LYY14}). Samples of the copula $$\hat{C}^{\bm{\gamma}^1, \cdots, \bm{\gamma}^m}$$ defined in eq. \eqref{def:newdmc} can be generated according to the following Algorithm \ref{alg:newdmc}.

\begin{algorithm}
\caption{Sampling algorithm of the DM copula \eqref{def:newdmc} generated by $\bar{ \bm \gamma}$}\label{alg:newdmc}
    \begin{algorithmic}[1]
        \Procedure{RobustDMC}{$\bm \alpha, \bar{ \bm \gamma},\mathfrak{C}, \mathfrak{D}$} 
            \State \textbf{sample} a random variable $Z^1$ distributing discretely as $\mathsf{P}(Z^1 =i) = \alpha_i$ for $i=0, \cdots, m$ 
            \If{$Z^1 \neq 0$}
                \State \textbf{sample} a random variable $Z^2$ distributing discretely as $\mathsf{P}(Z^2 = j | Z^1) = \gamma_j^{Z^1}$ for $j=1, \cdots, K$  
            \Else{ \textbf{set} $Z^2 = Z^1 = 0$} 
            \EndIf
            \State \textbf{sample} a random vector $\bm V = (V_1, \cdots, V_d)$ from the joint distribution $C_{Z^2}$
            \For{$k=1$ \textbf{to} d}
                \State $U_k = D_{Z^1 k}^{-1}(V_k)$
            \EndFor
            \State \textbf{return} $\bm U = (U_1, \cdots, U_d)$ 
        \EndProcedure
    \end{algorithmic}
\end{algorithm}
\noindent Samples of $$(X_1, \cdots, X_d)$$ with copula $\hat{C}^{\bm{\gamma}^1, \cdots, \bm{\gamma}^m}$ and  arbitrary  marginal distributions $F_1, \;F_2, \; \dots, \; F_d$ can be generated according to the quantile transformation
  $$
        (X_1, \cdots, X_d) \; =^d \; \left( F_1^{-1} (U_1), \cdots,  F_d^{-1} (U_d) \right).
$$
\subsubsection{Aggregate Loss}\label{subsec:agloss}

The simulation of the aggregate losses $X^{\bar{ \bm \gamma}}$ is now based on a simple transformation. Setting
$$A(s) \; =  \; \left\{(u_1, \cdots, u_d): \Psi\left( F_1^{-1}(u_1),\cdots, F_d^{-1}(u_d)\right) \leq s \right\}, $$
we define distribution functions
\begin{eqnarray*}
G_0(s)  & = &  \int \textbf{1}_{A(s)} d C_0(D_{01}(u_1), \cdots, D_{0d}(u_d)), \\   
G_{ij}(s) & = & \int \textbf{1}_{A(s)} d C_j(D_{i1}(u_1), \cdots, D_{id}(u_d)), \quad i=1, \cdots, m, \; j=1, \cdots, K,
\end{eqnarray*}
and note that 
 \begin{equation}\label{eq:mixturedist}
        F_{X^{\bar{\bm \gamma}} }(s)  \; :=   \;   \mathsf{P}[ X^{\bar{\bm \gamma}} \leq s]   \;  =  \; \alpha_0 G_0(s) +  \sum_{i=1}^m \alpha_i \sum_{j=1}^K \gamma_j^i G_{ij}(s).
    \end{equation}
If  $\Psi^0$ and $\Psi^{ij}$ are distributed according to $G_0$ and $G_{ij}$, $i=1, \cdots, m, \; j=1, \cdots, K$, and  independent of $Z_1$ and $Z_2$ defined in Algorithm \ref{alg:newdmc}, then 
    \begin{equation*}
        X^{\bar{\bm \gamma}} \; =^d \; \textbf{1}_{[Z_1=0]} \Psi^0 + \sum_{i=1}^m \sum_{j=1}^K \textbf{1}_{[Z_1=i, Z_2 = j]} \Psi^{ij}. 
    \end{equation*}
This representation will be instrumental for our simulation algorithms. For later use, we denote the density functions of $G_0(s)$, $G_{ij}(s)$, and $F_{X^{\bar{\bm \gamma}} }(s)$ by $g_0(s)$, $g_{ij}(s)$, and $f_{X^{\bar{\bm{\gamma}}}}(s)$, $i=1, \cdots, m, \; j=1, \cdots, K$, respectively, provided that they exist.

\section{Optimization Solver}\label{sec3:sto_opt}

In this section, we develop an algorithm solving problem \eqref{eq:opt3} that builds on two classical approaches: \emph{Stochastic Approximation (SA)} and \emph{Sample Average Approximation (SAA)}. While SA is an iterative optimization algorithm that is based on noisy observations, SAA first estimates the whole objective function and transforms the optimization into a deterministic problem. We combine both approaches.

The standard stochastic gradient algorithm of SA quickly approximates the worst-case risk, but the convergence to a worst-case copula is slow. It turns out that in many cases the risk is insensitive to certain directions in the mixture space of copulas. We exploit this observation in order to reduce the dimension of the problem and identify  a suitable subset of $\mathfrak{C}$ that excludes copulas whose contribution to the worst-case risk is small. We then determine a solution in the corresponding  simplex, relying on SAA, which is computationally efficient in lower dimensions only, but provides a good global solution to optimization problems, even if stochastic gradient algorithms are noisy and slow.

Our method thus first applies SA to estimate worst-case risk together with a candidate mixture from which a lower-dimensional problem is constructed. Second, SAA is used, but only in the lower-dimensional mixture space -- utilizing a large sample set that reduces noise. 

\paragraph{Step 1 -- Sampling.}
We generate $N$ independent copies of the $m\times K+1$ random variables
$\Psi^0$ and $\Psi^{ij}$
according to the distribution functions $G_0$ and $G_{ij}$, $i=1,\cdots,m,$ $j=1,\cdots,K $, respectively.

\paragraph{Step 2 -- SA Algorithm.} The \texttt{PSG-RobustAV@R} Algorithm \ref{alg:psg-cvar} discussed in Section~\ref{sec3:sa} seeks a candidate solution and terminates after a small number of iterations.  We design a stopping rule that determines when to move to the next step.

\paragraph{Step 3 -- SAA Algorithm.} 
From the solution obtained in Step 2 we construct a lower-dimensional simplex in which we search for a solution. We apply SAA on a suitable grid. The \texttt{SAA-RobustAV@R} Algorithm \ref{alg:saa-cvar} is discussed in Section~\ref{sec3:saa}. 

\subsection{Stochastic approximation: gradient approach} \label{sec3:sa}

SA is a recursive procedure evaluating noisy observations of the objective function and its subgradient. 
The algorithm moves in the gradient direction approaching a local optimum by a first-order approach (minimization and maximization require, of course, opposite signs).

\begin{algorithm} 
\caption{The projected stochastic gradient algorithm for the robust AV@R}\label{alg:psg-cvar}
    \begin{algorithmic}[1]
        \Procedure{PSG-RobustAV@R}{} 
            \State \textbf{Input} the  level $p$ of AV@R, the step sizes $\{ \delta_t = t^{-a} \}_{t \ge 1}$, the sample size sequences $\{N_t\}_{t \ge 1}$, the number of iterations $M$, 
            the PDF of $X^{\bar{\bm \gamma}_{t}}$ at iteration $t$ denoted by $f_{X^{\bar{\bm \gamma}_{t}}}(s)$,
            and the PDFs $g_0(s)$ and $g_{ij}(s)$ for $i=1, \cdots, m$ and $j=1, \cdots, K$
        
            \State \textbf{Initialization:}
            \State Set a starting state $ \bar{\bm \gamma}_{1} = \left( \bm{\gamma}^1_{1}, \cdots, \bm{\gamma}^m_{1} \right)$ with $\bm{\gamma}^i_{1} \in \bigtriangleup^{K-1}, i=1, \cdots, m$ 
            \While{terminal conditions are not met} 
                \For{$t=1$ \textbf{to} $M$}
                    \State \textbf{Simulation:} \Comment{generate $(L_1, \cdots, L_N)$ $N=N_t$ i.i.d. observations of  $X^{\bar{\bm \gamma}_{t}}$}
                    \For{$l=1$ \textbf{to} $N_{t}$}
                        \State Sample $(U_1^l, \cdots, U_d^l)$ from Algorithm \ref{alg:newdmc} with $\bar{\bm \gamma}_{t} = \left( \bm{\gamma}^1_{t}, \cdots, \bm{\gamma}^m_{t} \right)$
                        \State Set $L_l = \Psi\left( F_1^{-1} (U_1^l), \cdots,  F_d^{-1} (U_d^l) \right)$
                    \EndFor
                    \State \textbf{VaR and AV@R Estimation:}
                    \State Set $\hat v_p^{N_t} = L_{\lceil N_t p\rceil:N_t}$
                    \State Set $\hat c_p^{N_{t}} = \hat v_p^{N_{t}} + \frac{1}{N_{t}(1-p)} \sum_{i=1}^{N_{t}} (L_i - \hat v_p^{N_{t}} )^+$
                    \State \textbf{AV@R Gradient Estimation:}
                    \State Set $f_{X^{\bar{\bm \gamma}_{t}}}(s) = \alpha_0 g_0(s) +  \sum_{i=1}^m \alpha_i \sum_{j=1}^K \gamma_{j, t}^i g_{ij}(s)$ \Comment{$\gamma_{j,t}^i$ is the j-th component of  $\bm{\gamma}^i_{t}$}
                    \For{$i=1$ \textbf{to} $m$, \textbf{and} $j=1$ \textbf{to} $K$}
                        \State Set $\Delta_{i,j}(\hat c_p^{N_t}) = \frac{1}{N_t (1-p)} \sum_{l=1}^{N_t}   \frac{\alpha_i \ g_{ij} (L_l)}{f_{X^{\bar{\bm \gamma}_{t}}}(L_l)} \left(L_l - \hat v_p^{N_t} \right)  \textbf{1}_{\left[L_l \ge \hat v_p^{N_t} \right]}$
                    \EndFor
                    \State \textbf{Parameter Update - Multiple Simplexes Projection} 
                    \For{$i=1$ \textbf{to} $m$}
                        \State Set $\Delta^i_t = (\Delta_{i,1}, \cdots, \Delta_{i,K})$
                        \State Update $\bm{\gamma}^i_{t+1} = \Pi_{\bigtriangleup^{K-1}} \left( \bm{\gamma}^i_{t} + \delta_{t} \Delta^i_{t} \right)$ by Algorithm \ref{alg:proj}
                    \EndFor
                \EndFor
            \EndWhile
            \State \textbf{Output} $\hat c_p^N$ and $\bm{\gamma}^1_{t}, \cdots, \bm{\gamma}^m_{t}$
        \EndProcedure
    \end{algorithmic} 
\end{algorithm}
\begin{algorithm}
\caption{Euclidean projection of a vector $y$ onto simplex}\label{alg:proj}
    \begin{algorithmic}[1]
        \Procedure{ProjS}{$\bm y$} \Comment{$\bm y \in \mathbb{R}^K$}
            \State sort $\bm y$ into $\bm u$: $u_1 \ge u_2 \ge \cdots u_K$
            \State find $\tau = \max\{ 1\le j \le K: u_j + \frac{1}{j} (1-\sum_{k=1}^j u_k) > 0 \}$
            \State define $\lambda = \frac{1}{\tau} (1-\sum_{k=1}^\tau u_k)$ 
            \State \textbf{return} $\bm x$  s.t. $x_i = \max (y_i+\lambda, 0), i=1, \cdots, K$ \Comment{$\bm x = \Pi_{\bigtriangleup^{K-1}} (\bm y) \in \mathbb{R}^K $}
        \EndProcedure
    \end{algorithmic}
\end{algorithm}

\subsubsection{Projected stochastic gradient method}

 Algorithm  \ref{alg:psg-cvar} seeks to solve the optimization problem \eqref{eq:opt3}. At each iteration $t$  the SA algorithm first generates $N= N_t$ loss samples $L_1, \cdots, L_{N_t}$ of $X^{\bar{\bm{\gamma}}_{t}}$ according to Algorithm \ref{alg:newdmc}. SA then estimates the V@R and AV@R as follows:
    \begin{eqnarray*}
        \hat v_p^{N_t} = L_{\lceil N_t p\rceil:N_t }, \quad
        \hat c_p^{N_t} = \hat v_p^{N_t} + \frac{1}{N_t (1-p)} \sum_{i=1}^{N_t} (L_i - \hat v_p^{N_t})^+.
    \end{eqnarray*}
Here, $\lceil a\rceil$ denotes the smallest integer larger than or equal to $a$, and $L_{s:N}$ is the s-th order statistic from the $N$ observations, $L_{1:N} \le L_{2:N} \le \cdots \le L_{N:N}$.

Second, SA computes the gradients $\Delta^i_t = (\Delta_{i,1}, \cdots, \Delta_{i,K})$ of $c_p$ at $\bm{\gamma}^i_{t}$ from     \begin{eqnarray}\label{eq:GCVaR}
        \Delta_{i,j}(\hat c_p^{N}) 
        &=& \frac{1}{N(1-p)} \sum_{l=1}^N   \frac{\partial \ \log f_{X^{\bar{\bm \gamma}}} (L_l)}{\partial \gamma_j^i} \left(L_l - \hat{v}_p^N \right)  \textbf{1}_{\left[L_l \ge \hat{v}_p^N \right]} \nonumber \\
        &=& \frac{1}{N(1-p)} \sum_{l=1}^N   \frac{\alpha_i \ g_{ij} (L_l)}{f_{X^{\bar{\bm \gamma}}} (L_l)} \left(L_l - \hat{v}_p^N \right)  \textbf{1}_{\left[L_l \ge \hat{v}_p^N \right]}
    \end{eqnarray}
 for every  $i=1, \cdots, m$. 

Third, parameter updates are computed for each $i$:
    \begin{equation} \label{eq:sa}
        \bm{\gamma}^i_{t+1} = \Pi_{\bigtriangleup^{K-1}} \left( \bm{\gamma}^i_{t} + \delta_{t} \Delta^i_{t} \right),
    \end{equation}
where $\Pi_{\bigtriangleup^{K-1}}(x) = \arg \min_y \{ ||x-y|| \ | \ y \in \bigtriangleup^{K-1}\}$ is the Euclidean projection of $x$ onto the simplex, and  $\{ \delta_{t} \}_{t \ge 1}$ is the step size multiplier. 
This type of algorithm is called \textit{the projected gradient descent algorithm}.

Algorithm \ref{alg:psg-cvar}, a projection onto multiple high dimensional simplexes, applies methods described in \ci{Con16}. In contrast to these, the simple, classical projection Algorithm \ref{alg:proj}  that we included for illustration possesses the larger complexity $O(K^2)$.

\subsubsection{Convergence of SA}

The convergence of Algorithm \ref{alg:psg-cvar} is guaranteed if Assumptions \ref{ass:grd} \& \ref{ass:stepsize} below are satisfied, see Theorem 5.2.1 in \ci{KY03}. 

\begin{assumption}\label{ass:grd}
    \begin{enumerate}
        \item[(1)] The random variables $X^{\bar{\bm \gamma}}$ have a continuous distribution with finite support for all $\bar{\bm \gamma}$.
        \item[(2)] For all $\bar{\bm \gamma}$, $i,j$, the gradients $\frac{\partial}{\partial \gamma_j^i} v_p(\bar{\bm \gamma})$ and $\frac{\partial}{\partial \gamma_j^i} c_p(\bar{\bm \gamma})$ are well defined and bounded.
        \item[(3)] $X^{\bar{\bm \gamma}}$ has a positive and continuously differentiable density $f_{X^{\bar{\bm \gamma}}}$, and $\frac{\partial}{\partial \gamma_j^i}  \log f_{X^{\bar{\bm \gamma}}}(s)$ exists and is bounded for all $s$, $\bar{\bm \gamma}$, $i,j$.
    \end{enumerate}
\end{assumption}

\begin{assumption}\label{ass:stepsize}
\begin{enumerate}
    \item[(1)]  The step size sequence $\{ \delta_{t} \}_{t \ge 1}$ satisfies 
    $$\sum_{t=1}^\infty \delta_{t} = \infty,~ \delta_{t} \ge 0,~ \sum_{t=1}^\infty \delta_{t}^2 < \infty.$$
    \item[(2)] $\frac{\partial}{\partial \gamma_j^i} c_p(\bar{\bm \gamma})$ is  continuous, and    \begin{equation*}
        \sum_{t=1}^\infty \delta_{t} \left| \mathsf{E}\left[\Delta_{i,j}\left(\hat c_p^{N_{t}}\right)\right] - \frac{\partial}{\partial \gamma_j^i} c_p\left(\bar{\bm \gamma}_{t}\right) \right| < \infty
    \end{equation*}
    with probability 1 for each $i$ and $j$.
\end{enumerate}
\end{assumption}

In specific applications, these sufficient conditions for convergence are not always satisfied. However, the SA algorithm might still produce estimates that approach a solution of the problem. We will impose a switching condition that determines when we move from SA to SAA. SAA is explained in the next section.

\begin{remark}[Concavity of AV@R]

Algorithm \ref{alg:psg-cvar} converges to a local maximum. But any local maximum of the problem \eqref{eq:opt3} is even the global maximum, since $c_p$ is a concave function of $\bar{\bm{\gamma}}$, see \ci{AS13}. This property is closely related to the m-concavity of AV@R, a concavity property on the level of distributions, see \ci{Web06} and \ci{BB15}.
\end{remark}

\begin{remark}[Differentiability of AV@R]
The gradient estimate \eqref{eq:GCVaR} of AV@R in Algorithm \ref{alg:psg-cvar} belongs to the Likelihood Ratio (LR) methods due to \ci{TGM15}. A LR approach is appropriate, since the distribution of the argument $X^{\bar{\bm \gamma}}$ of the AV@R depends on $\bar{\bm \gamma}$.
\begin{enumerate}
\item The computation \eqref{eq:GCVaR} needs $g_{ij}$ as inputs. If their computation is not analytically tractable, an empirical estimator can be chosen. Other options are AEP (\ci{AEP11}) and  GAEP (\ci{AEP12}).
\item An alternative to LR gradient estimation are finite differences, as applied in algorithms of  Kiefer–Wolfowitz type (\ci{KW52}). Properties of such algorithms are discussed in \ci{BCG11}. Finite differences require less regularity in order to be applicable, but typically exhibit a worse performance.
\end{enumerate}
\end{remark}

\subsection{Sample average approximation} \label{sec3:saa}

 AV@R belongs to the class of divergence risk measures that coincide, up to a sign change, with optimized certainty equivalents. These admit a representation as the solution of a one-dimensional optimization problem, see \ci{BT07}. The minimizer can be characterized by a first order condition. For the specific case of AV@R this representation was previously described in \ci{Pflug00}, \ci{RU00}, and \ci{RU02}, and implies the following identity:
 $$ c_p (\bar{\bm{\gamma}}) = \min_{ u \in \mathbb{R} } \left\{ u + \frac{1}{1-p} \int \left(L - u\right)^+ d  F_{X^{\bar{\bm{\gamma}}}}(L) \right\}.$$ 
The mixture representation \eqref{eq:mixturedist} of the distribution function of $X^{\bar{\bm{\gamma}}}$ provides a  reformulation of the original problem \eqref{eq:opt3}:
    \begin{align}\label{opt:saa}
        &\max_{\bar{\bm{\gamma}}  \in (\bigtriangleup^{K-1})^m} \min_{ u \in \mathbb{R} }  \left\{ u + \frac{\alpha_0}{1-p} \mathsf{E}[\Psi^0- u]^+ +
        \sum_{i=1}^m \sum_{j=1}^K \frac{\alpha_i \gamma_j^i}{1-p} \mathsf{E}[\Psi^{ij}- u]^+\right\}
    \end{align}
where $\Psi^0$ and $\Psi^{ij}$ are random variables with distributions $G_0$ and $G_{ij}$, respectively. 

SAA algorithmically solves the stochastic optimization problem \eqref{eq:opt3} by first approximating the objective function by its sample average estimate and then solving the auxiliary deterministic problem.
Eq.~\eqref{opt:saa} suggests the following SAA for \eqref{eq:opt3}:
    \begin{align}\label{opt:saa2}
        &\max_{\bar{\bm{\gamma}}  \in (\bigtriangleup^{K-1})^m} \min_{ u \in \mathbb{R} }  \left\{ u + \frac{\alpha_0}{1-p} \frac{1}{N} \sum_{k=1}^N [\Psi^0_k- u]^+ + \sum_{i=1}^m \sum_{j=1}^K \frac{\alpha_i \gamma_j^i}{1-p} \frac{1}{N} \sum_{k=1}^N [\Psi^{ij}_k- u]^+\right\}
    \end{align}
The SAA procedure is described in Algorithm \ref{alg:saa-cvar}.

\begin{algorithm} 
\caption{SAA for the robust AV@R}\label{alg:saa-cvar}
    \begin{algorithmic}[1]
    \Procedure{SAA-RobustAV@R}{} 
        \State \textbf{Input} the probability level $p$ for AV@R, $\{ \Psi^0_k, \Psi^{ij}_k \}_{k=1, \cdots, N}$ $N$ realizations $\Psi^0$ and $\Psi^{ij}$ for $i=1, \cdots, m$ and $j=1, \cdots, K$, $\bm \alpha$, a grid of $\bar{\bm{\gamma}} = (\bm{\gamma}^1_1, \cdots, \bm{\gamma}^m_1)$
        \For{ every $\bar{\bm{\gamma}} = (\bm{\gamma}^1_1, \cdots, \bm{\gamma}^m_1)$ on the grid}
            \State \textbf{Initialization:}
            Initialize the lower bound $u_l$ with $\bar{p}_N(u_l) < p$, and the upper bound $u_u$ with $\bar{p}_N(u_u) > p$ 
            \State Set $p_u = \bar{p}_N(u_u)$ and $p_l = \bar{p}_N(u_l)$
            \State \textbf{Bisection method:}
            \While{$|p_u - p| > \epsilon$ and $|p_l - p| > \epsilon$} 
                \State $u_m = (u_u + u_l)/2$ and evaluate $p_m = \bar{p}_N(u_m)$
                \If{$p_m > p$} set $u_u = u_m$ and $p_u = p_m$
                \Else ~set $u_l = u_m$, and $p_l = p_m$
                \EndIf
            \EndWhile
            \If{$|p_u - p| \le \epsilon$} return $\check{u}^N(\bar{\bm{\gamma}}) = u_u$
            \Else ~return $\check{u}^N(\bar{\bm{\gamma}}) = u_l$
            \EndIf
            \State \textbf{AV@R computation:} 
            \State Compute $\check{c}_p^N(\bar{\bm{\gamma}}) = \check{u}^N + \frac{\alpha_0}{1-p} \frac{1}{N} \sum_{k=1}^N [\Psi^0_k- \check{u}^N]^+ + \sum_{i=1}^m \sum_{j=1}^K \frac{\alpha_i \gamma_j^i}{1-p} \frac{1}{N} \sum_{k=1}^N [\Psi^{ij}_k- \check{u}^N]^+$
        \EndFor
        \State \textbf{Output} $\max  \check{c}_p^N(\bar{\bm{\gamma}})$ on the grid
    \EndProcedure
    \end{algorithmic} 
\end{algorithm}

\subsubsection{Inner minimization} 

The inner minimization in \eqref{opt:saa} can numerically be solved on the basis of first order conditions that are specified in the following lemma. 

\begin{lemma}\label{lem:zeta}
    Let $\zeta(u,\bar{\bm{\gamma}}) = u + \frac{\alpha_0}{1-p} \mathsf{E}[\Psi^0- u]^+ +
    \sum_{i=1}^m \sum_{j=1}^K \frac{\alpha_i \gamma_j^i}{1-p} \mathsf{E}[\Psi^{ij}- u]^+$. 
    Then 
    \begin{eqnarray*}
        \frac{\partial^+\zeta}{\partial u}(u,\bar{\bm{\gamma}}) = \frac{-p}{1-p} + \frac{\alpha_0}{1-p} \mathsf{P}(\Psi^0 \le u) + \sum_{i,j} \frac{\alpha_i \gamma_j^i}{1-p} \mathsf{P}(\Psi^{ij} \le u), \\
        \frac{\partial^-\zeta}{\partial u}(u,\bar{\bm{\gamma}}) = \frac{-p}{1-p} + \frac{\alpha_0}{1-p} \mathsf{P}(\Psi^0 < u) + \sum_{i,j} \frac{\alpha_i \gamma_j^i}{1-p} \mathsf{P}(\Psi^{ij} < u). 
    \end{eqnarray*}
The minima of the function $u \mapsto \zeta(u,\bar{\bm{\gamma}})  $ are attained and any minimizer $z$ satisfies
    \begin{equation}\label{foc1}
        \frac{\partial^-\zeta}{\partial u}(u,\bar{\bm{\gamma}}) \le 0 \le \frac{\partial^+\zeta}{\partial u}(u,\bar{\bm{\gamma}}).
    \end{equation}
  If the distribution functions of  $\Psi^{0}$ and $\Psi^{ij}$ are continuous, the first order condition \eqref{foc1} becomes
    \begin{equation}\label{foc2}
        p = \alpha_0 \mathsf{P}(\Psi^0 \le u) + \sum_{i,j} \alpha_i \gamma_j^i \mathsf{P}(\Psi^{ij} \le u).
    \end{equation}
\end{lemma}
\begin{proof}
See Appendix A2.
\end{proof}
Replacing  $\mathsf{P}(\Psi^0 \le u)$ and $\mathsf{P}(\Psi^{ij} \le u)$ in \eqref{foc1} and \eqref{foc2} by the empirical probabilities, we obtain a SAA approach to solve the root finding problems posed by the first order conditions. The sample version  of \eqref{foc2} is  $$
        \bar{p}_N(u) = \alpha_0 \frac{1}{N} \sum_{k=1}^N \textbf{1}_{[\Psi^0_k \le u]}  + \sum_{i=1}^m \sum_{j=1}^K \frac{\alpha_i \gamma_j^i}{N} \textbf{1}_{[\Psi^{ij}_k \le u]}.$$
Utilizing a simple bisection method, one can determine the root $\check{u}^N(\bar{\bm{\gamma}})$ that solves $\bar{p}_N(u) = p$.

\subsubsection{Outer maximization} 

The sample version of the outer maximization in \eqref{opt:saa} is 
$$ \max_{\bar{\bm{\gamma}}} \left\{ \underbrace{ \check{u}^N(\bar{\bm{\gamma}}) + \frac{\alpha_0}{1-p} \frac{1}{N} \sum_{k=1}^N [\Psi^0_k- \check{u}^N(\bar{\bm{\gamma}})]^+ + \sum_{i=1}^m \sum_{j=1}^K \frac{\alpha_i \gamma_j^i}{1-p} \frac{1}{N} \sum_{k=1}^N [\Psi^{ij}_k- \check{u}^N(\bar{\bm{\gamma}})]^+}_{=: \; \check{c}_p^N(\bar{\bm{\gamma}})} \right\} $$
Algorithm \ref{alg:saa-cvar} evaluates $\check{c}_p^N(\bar{\bm{\gamma}})$ for all $\bar{\bm{\gamma}}$ on a grid, compares the values of the function and thereby determines an approximate solution. 

\subsubsection{Switching condition and dimension reduction}\label{sec:switching}

The outer maximization requires the computation at many grid points and is expensive in high dimensions. We propose to identify a suitable lower-dimensional subsimplex in the space of copulas on the basis of SA, before we switch to SAA. This is justified by the fact that the worst-case risk is typically insensitive to contributions of some of the copulas in $\mathfrak{C}$.  Before we summarize the full procedure, we address the switching condition from SA to SAA.

SA produces a random sequence $(\bar{\bm{\gamma}}_t)_{t=1,2, \dots}$. We choose a certain burn-in period $t_{\min}$ and a maximal number of SA-iterations $t_{\max}$ to construct a stopping time $t^* \in \{t_{\min}, t_{\min} + 1, \dots, t_{\max} \}$. We stop at $t$ when two consecutive matrices $\bar{\bm{\gamma}}_{t-1}$ and $\bar{\bm{\gamma}}_{t}$ are close to each other according to some metrics. In the examples below, we implement the 1-norm $\|A\| = \sum_i \sum_j  | A_{ij} |$ and a threshold level of $0.01$. Moreover, we choose $t_{\min}=10$ and $t_{\max} = 50$.

When switching to SAA, the dimension of the problem is reduced as follows.   To simplify the notation, we write $$\bar{\bm{\gamma}}= 
\begin{pmatrix}
\gamma_1^1 & \gamma_1^2 & \cdots & \gamma_1^m  \\
\gamma_2^1 & \gamma_2^2 & \cdots & \gamma_2^m  \\
 \vdots& \vdots &   \ddots & \vdots   \\
\gamma_K^1 & \gamma_K^2 & \cdots & \gamma_K^m  \\
\end{pmatrix}
$$ instead of $\bar{\bm{\gamma}}_{t^*}$ where $t^*$ is the stopping time described above. Recall that the index $j=1, 2, \dots, K$ enumerates the copulas in $\mathfrak{C}$, while $i=1,2, \dots, m$ labels the weights $\alpha_i$ and corresponding distortions $D_{i1}, D_{i2}, \dots, D_{id}$ in eq.~\eqref{def:dmc} or eq.~\eqref{def:newdmc}, respectively. We assume that the weights are equal, i.e., $\alpha_i= \alpha$ $\forall i=1,2, \dots, m$; this assumption ensures that the probability fraction of the total dependence that is governed by each column of $\bar{\bm{\gamma}}$ is equal for all columns. 

We first select the number of copulas $K^*< K$ we wish to select from $\mathfrak{C}$ for the application of SAA. We distinguish the cases $K^* \leq m$ and $K^*>m$. In the first case, we identify the largest entry\footnote{If there is a tie, we select the one with the larger gradient.} from $\bar{\bm{\gamma}}$ and select the copula corresponding to it. We remove the corresponding row and the corresponding column from $\bar{\bm{\gamma}}$, identify the largest entry from the remaining matrix, and remove again the corresponding row and column. We proceed iteratively until $K^*$ copulas are selected. In the second case, i.e., $K^*>m$, all rows and columns are removed from $\bar{\bm{\gamma}}$, after $m$ copulas were selected. In this case, we proceed with selecting copulas $m+1, m+2, \dots$ as follows. We remove all rows corresponding to the $m$ copulas that were already selected, and then proceed in the same manner as described above to select the remaining copulas.

\begin{remark}
For each $i=1,2, \dots, m$, the mixture copula corresponding to ${\bm{\gamma}}^i$ governs a probability fraction $\alpha$ of the overall dependence structure in a region determined by the distortions $D_{i1}, D_{i2}, \dots, D_{id}$. The algorithm consecutively selected for different $i$ the most important element from the copulas in $\mathfrak{C}$ that were not previously selected. This guarantees that the contributions of the vectors of distortion functions corresponding to different values of $i$ are taken into consideration when the $K^*$ copulas are chosen.
\end{remark}

\subsection{Full procedure}\label{subs:fp}

We finally present a brief summary of our proposed algorithm. 
\paragraph{Step 1 -- Sampling.}
    \begin{enumerate}
        \item[1.]  Generate $N = N_t$ independent copies of the random variables $\Psi^0$, $\Psi^{ij}$ as described in Section~\ref{subs:Sampl}.      
        \item[2.]  Use the samples to estimate the densities $g_0(s)$, $g_{ij}(s)$, and $f_{X^{\bar{\bm{\gamma}}}}(s)$. The values are stored for the inter- and extrapolation.
        \item[3.]  If necessary, generate new samples according to an importance sampling density $h$.
    \end{enumerate}

\paragraph{Step 2 -- SA Algorithm.}
    \begin{enumerate}
        \item[4.] Apply \texttt{PSG-RobustAV@R} Algorithm \ref{alg:psg-cvar}. 
        Importance sampling techniques can be adopted as illustrated in Section \ref{sec:ex} below --  if applicable. 
        \item[5.] If the switching condition described in Section~\ref{sec:switching} is met, terminate the algorithm and determine a selection of the most important copulas in order to reduce the dimensionality of the problem.\\
    \end{enumerate}

\paragraph{Step 3 -- SAA Algorithm.} 
    \begin{enumerate}
        \item[6.] Construct a suitable grid on the lower-dimensional simplex. Adaptively refine the grid in a smaller domain on the basis of the results of the application of the algorithm specified in the next step, and apply the algorithm again on the new grid.
        \item[7.] Apply \texttt{SAA-RobustAV@R} Algorithm \ref{alg:saa-cvar} to find the worst case over grid points. The worst case is the estimated solution to the original problem \eqref{eq:opt3} on the lower-dimensional mixture space of copulas chosen in Task 5.   
    \end{enumerate}

\subsection{A motivating example}\label{sec:ex}
Before we discuss applications to finance markets and cyber risk, we illustrate our procedure in the context of a simple example motivated by \ci{LYY14}. 

\begin{example}[$m=2, K=5, d=2$]\label{ex1}
    Consider aggregate losses $X = X_1 + X_2 $ with individual losses $X_1, X_2 \in L^1$.  The distributions of the individual positions are inverse Gaussian with density 
    $$ x \; \mapsto \;  \sqrt{\frac{\lambda}{2\pi x^3}} \exp \left( -\frac{\lambda}{2\mu^2x} (x-\mu)^2 \right).$$ The dependence of the positions is uncertain, and we would like to evaluate the worst-case AV@R at level $p \in (0,1)$. Letting  $\alpha_0 = 1-2\alpha$ and $\alpha_1 = \alpha_2 = \alpha$ with $\alpha = 0.1$, we assume that $D_i = D_{i1} = \cdots = D_{id}$ for all $i=0,1,2$ and choose the distortion functions 
    \begin{equation}\label{ex:dst}
        D_1(x) = \frac{x - \alpha x^2}{\alpha + (1-2\alpha)x}, \; D_2(x) = \frac{\alpha x^2}{\alpha + (1-2\alpha)(1-x)}, \; D_0(x) = \frac{x - \alpha D_1(x) - \alpha D_2(x)}{1-2\alpha}.
    \end{equation} 
   The copula capturing dependence in the typical area is modeled by a d-dimensional Gaussian copula $$C_0 = C_\Sigma^{Ga} = \Phi_{\bm{0},\Sigma}(\Phi^{-1}(u_1), \cdots, \Phi^{-1}(u_d));$$ here, $\Phi$ and $\Phi_{\bm{0},\Sigma}$ signify the standard univariate normal distribution function and the multivariate Gaussian distribution function with covariance matrix $\Sigma$, respectively. 
   
The  family of copulas $\mathfrak{C}= \{ C_1, C_2, C_3, C_4, C_5 \}$ is specified as follows:    
\begin{itemize}
        \item[$C_1$:] a t-copula  $C^{t}_{\nu,P}\left(u_1, \cdots, u_d) = \bm t_{\nu,P}(t_\nu^{-1}(u_1), \cdots, t_\nu^{-1}(u_d)\right)$ where $t_\nu$ is a standard univariate t distribution with $\nu$ degree of freedom and $\bm t_{\nu,P}$ is the joint distribution with a correlation matrix $P$;
        \item[$C_2$:] a Clayton copula $C^{Cl}(u_1, \cdots, u_d)  = \left( \sum_{i=1}^d u_i^{-\theta} - d + 1 \right)^{- 1/\theta},~ 0 < \theta < \infty$;
        \item[$C_3$:] a Gumbel copula $C^{Gu}(u_1, \cdots, u_d)  = \exp\left\{ - \left[ \sum_{i=1}^d (-\log u_i)^\theta \right]^{1/\theta}\right\},~ 1 \leq \theta < \infty$; 
        \item[$C_4$:] a Frank's copula $C^{Fr}(u_1, \cdots, u_d)  = \log_\theta \left\{ 1 + \frac{\prod_{i=1}^d (\theta^{u_i} - 1)}{(\theta-1)^{d-2}} \right\},~ \theta \ge 0$; 
        \item[$C_5$:] the independence copula $\Pi(u_1, \cdots, u_d) = \prod_{i=1}^d u_i$.
    \end{itemize}
 
\end{example}

\subsubsection*{SA Algorithm}

Step 2 in the full procedure summarized in Section~\ref{subs:fp}  is the SA Algorithm \ref{alg:psg-cvar}. Its step size is given by $t^{-a}$ for $0.5 < a \leq 1$. Figure \ref{fig:perf_ait_gamma} illustrates the dynamics of the corresponding weights in the random sequence $(\bar{\bm \gamma}_t)_t$ for the five copulas in $\mathfrak{C}$ and the distortions $D_1$ and $D_2$ for a specific numerical example. We vary the step size and compare $a = 0.6, 0.7, 0.8, 0.9$ for the first 200 iterations. The approximation becomes faster for smaller $a$. 

The downside risk measures by AV@R is mainly governed by the upper tail of the losses whose dependence structure is encoded by the distortion function $D_2$. This is captured by the second column in Figure \ref{fig:perf_ait_gamma} which shows that the weights of copulas $C_2$ (Clayton copula), $C_4$ (Frank's copula), and $C_5$ (independence copula) decrease quickly to zero. The maximal AV@R is mainly determined by $\gamma_1^2$ (the weight of t-copula $C_1$ for the upper tail, $D_2$) and $\gamma_3^2$ (the weight of Gumbel copula $C_3$ for the upper tail, $D_2$). These observations suggest that dimension reduction as described in Section~\ref{sec:switching} can successfully be implemented for this example.

\begin{figure}[h]
    \centering
    \includegraphics[width=\textwidth]{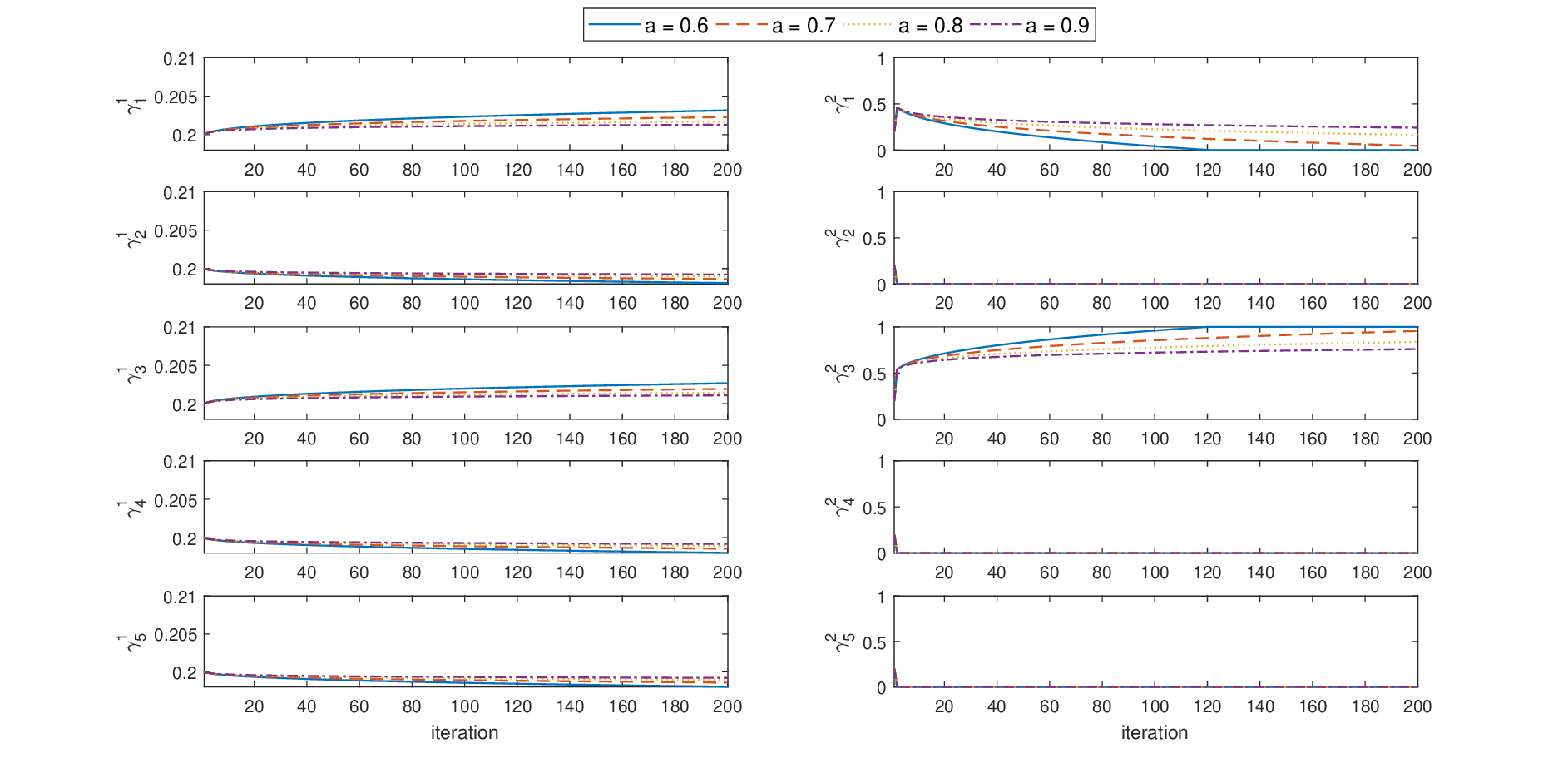}
    \caption{\label{fig:perf_ait_gamma} SA-results for varying $a \in (0.5,1]$ with  step sizes $\{ \delta_{t} \} = \{t^{-a} \}_{t \ge 1}$ in Example \ref{ex1}. The off-diagonal elements of $\Sigma$ equal  $0.7$, the diagonal elements $1$. We set $\nu = 1$, $\theta = 0.7565, 1.7095, 1.2$ for $C_2, C_3, C_4$, respectively. The IG parameters are $\mu_1 = \mu_2 = 1, \lambda_1 = 0.5, \lambda_2 = 1.2$. The sample size is fixed as $N_t= 10^6$ for every iteration $t$ and a kernel density at $1000$ equally spaced points is used based on $5\times 10^7$ sample data. 
}   
\end{figure}

The initial AV@R at $p=0.99$ for uniform\footnote{All entries of the matrix are equal.} $\bar{\bm \gamma}_1$ is reported as $13.8657$ for $a=0.6$, while AV@R is increased to $14.6832$ just after five iterations. In fact, this number is hardly distinguishable from the estimated optimal value found in SAA later on. We observe that AV@R values become insensitive to changes in  $\bar{\bm{\gamma}}_t$ after just a few iterations. This observation provides further motivation for the suggested approach  to reduce the dimension of the problem (see Section~\ref{sec:switching}).  

\subsubsection*{Importance sampling}

We explore the potential to reduce the variance of the estimators by an application of importance sampling. Recall the notation introduced in Section~\ref{subsec:agloss}. If $h$ is a density that dominates $f_{X^{\bar{\bm \gamma}}}$, we may sample from $h$ and modify Algorithm \ref{alg:psg-cvar} to obtain importance sampling estimators replacing (i) VaR $\hat v_p^N$, (ii) AV@R $\hat c_p^N$, and (iii) the AV@R gradient  $\Delta_{i,j}(\hat c_p^{N})$.

Letting $\mathcal{L} = \frac{f_{X^{\bar{\bm \gamma}}}}{h} $ be the likelihood ratio, we estimate the corresponding IS empirical distribution $\tilde {F}^{IS}_{X^{\bar{\bm \gamma}}}(s)$ by
$$ \tilde {F}^{IS}_{X^{\bar{\bm \gamma}}}(s) = \frac{1}{N} \sum_{l=1}^N \mathcal{L}(y_l) \textbf{1}_{[y_l \le s]} , \quad s \in \bbr,$$
with $y_l $  drawn iid from $h$. The corresponding IS estimators are 
    \begin{eqnarray*}
        && \tilde v_p^N = \inf \{ s: \tilde {F}^{IS}_{X^{\bar{\bm \gamma}}}(s) 
         \ge p \};~ \tilde c_p^N = \tilde v_p^N + \frac{1}{N(1-p)} \sum_{l=1}^N (y_l - \tilde v_p^{N})^+ \mathcal{L}(y_l) ;\\
        && \Delta_{i,j}(\tilde c_p^{N}) = \frac{1}{N(1-p)} \sum_{l=1}^{N}\frac{\alpha_i \ g_{ij} (y_l)}{f_{X^{\bar{\bm \gamma}}} (y_l)} \left( y_l - \hat v_p^{N} \right) \mathcal{L}(y_l)   \textbf{1}_{\left[y_l \ge \tilde v_p^{N} \right]}.
    \end{eqnarray*}
 
Motivated by eq.~\eqref{eq:mixturedist}, we propose to define the IS density $h$ as a mixture that relies on measure changes  of the distribution functions $G_0$, $G_{ij}$ with densities $g_0$, $g_{ij}$, $i=1,2, \dots, m$, $j=1,2, \dots, K$:
    \begin{equation*}\label{density_h}
        h(x) = \alpha_0 h_0(x) +  \sum_{i=1}^m \alpha_i \sum_{j=1}^K \gamma_j^i h_{ij}(x).
    \end{equation*}
For simplicity, we modify only two ingredients: 

We replace the central copula $C_0$ by an importance sampling copula $\tilde C_0$ and the marginal distributions $F_i$ by importance sampling distributions $\tilde{F}_i$; all other ingredients of the family of joint distributions of $X_1, X_2, \dots, X_d$ in Example~\ref{ex1}, in particular the collection $\mathfrak{C}$, are not changed. We thus obtain the following identities: 
\begin{eqnarray*}
    h_0(x) &=& \frac{\partial}{\partial x} \int \textbf{1}_{\tilde A(x)}\ d \tilde C_0(D_{01}(u_1), \cdots, D_{0d}(u_d)); \\
    h_{ij}(x) &=& \frac{\partial}{\partial x} \int \textbf{1}_{\tilde A(x)}\ d C_j(D_{i1}(u_1), \cdots, D_{id}(u_d)) \quad \quad \forall i,j; \\
    \tilde A(x) &=& \left\{(u_1, \cdots, u_d): \Psi\left( \tilde F_1^{-1}(u_1),\cdots, \tilde F_d^{-1}(u_d)\right) \leq x \right\}.
\end{eqnarray*}
Many other strategies to design IS distributions are, of course, possible. However, good IS methodologies for copulas are challenging. At the same time, the total computational effort must be estimated in order to evaluate the overall efficiency of competing algorithms. These issues constitute an interesting topic for future research.

To illustrate the potential of IS, we consider Example \ref{ex1}. As suggested by \ci{HSX10}, we shift the mean vector of the Gaussian copula $C_0$ to obtain $\tilde C_0$. On the marginal distributions, we utilize for each $i=1,2, \dots, m$ an Esscher measure change with parameter $w_i$ that transforms an  inverse Gaussian distribution $IG(\mu_i, \lambda_i)$ to a shifted IG distribution
$\tilde{F}_i \sim IG\left(\frac{\mu_i\sqrt{\lambda_i}}{\sqrt{\lambda_i - 2\mu_i^2 w_i}}, \lambda_i\right)$ with  $w_i \le \frac{\lambda_i}{2\mu_i^2}$ .

Numerical results display significant variance reduction. For example, in a typical case study with $10^6$  samples the variance of the crude MC estimator of robust AV@R is $0.00247$  while the importance sampling variances are reported as $0.00044$  and $0.00042$ for the historical likelihood estimator and the kernel estimator, respectively. 
We observe average variance reduction ratios around $5$ to $7$ across samples with the following new parameters: for exponential tilting $w_1 = 0.1$ (new $\mu_1^{IG} =1.2910$), $w_2 = 0.3$ (new $\mu_2^{IG} =1.4142$) and for the shifted drift for Gaussian distribution $\mu_1^G = 0.5$, and $\mu_2^G = 1$.  

\subsubsection*{Switching to SAA}

We apply the methodology described in Section~\ref{sec:switching}.  Setting $t_{\min}=10$ and $t_{\max}=50$, we run SA with uniform initial values, i.e., all entries of $\bar{\bm{\gamma}}_1$ are $1/5$,  and with a sample size $N_t = 10^5$ for step size $a=0.6$. Recall Algorithm~\ref{alg:psg-cvar} for a description of the parameters. The initial choice of $\bar{\bm{\gamma}}$ corresponds to an AV@R at level 0.99 of 13.8046. This result differs slightly from the initial value reported in Figure~\ref{fig:perf_ait_gamma} due to sampling error. 

The stopping time equals $t^*=17$ with corresponding 
$$\bar{\bm{\gamma}}^\top \;  = \;
    \begin{pmatrix}
    0.2008 & 0.1994 & 0.2007 & 0.1994 & 0.1994\\
    0.3309 & 0 & 0.6690 & 0 & 0 
    \end{pmatrix}$$
   and AV@R at level $0.99$ of $14.8094$ with an empirical standard deviation of $0.1326$ computed from the last ten iterations.
 The increments of the sequence $(\bar{\bm{\gamma}}_t)_{t=1,2, \dots}$ are already small at the stopping time $t^*$:
 $$ (\bar{\bm{\gamma}}_{17} - \bar{\bm{\gamma}}_{16} )^\top
\;  = \;  
 \begin{pmatrix}
    0.00002 & -0.00001 & 0.00001 & -0.00001 & -0.00001\\
    -0.004 & 0 & 0.004 & 0 & 0 
 \end{pmatrix}.$$
For comparison, at iteration $100$ we obtain  a corresponding
$$\bar{\bm{\gamma}}^\top \;  = \;
    \begin{pmatrix}
      0.2017 & 0.1989 & 0.2015 & 0.1988 & 0.1988\\
    0.1484 & 0 & 0.8515 & 0 & 0 
    \end{pmatrix}$$
and AV@R at level $0.99$ of $14.7715$ with an empirical standard deviation of $0.1553$ computed from the last ten iterations.
These observations indicate that SA quickly approximates the worst-case AV@R. However, the precision improves only very slowly afterwards. The convergence to the optimal value of $\bar{\bm{\gamma}}$ is slow for some components.

In order to reduce the dimension of the problem according to Section~\ref{sec:switching}, we set $K^*=2$ and select for the application of SAA the copulas $C_1$ (t-copula) and $C_3$ (Gumbel copula) on the basis of the estimate $\bar{\bm{\gamma}}_{17}$. Thus, SAA needs to be applied to a two-dimensional grid for
$$ \bar{\bm{\gamma}}^\top = 
\begin{pmatrix}
     \gamma_1^1 &  0 & \gamma_3^1 & 0 & 0\\
    \gamma_1^2 &  0 & \gamma_3^2 & 0 & 0
    \end{pmatrix}, \quad\quad \gamma_1^1 + \gamma_3^1 =1, \;  \gamma_1^2 + \gamma_3^2 =1,\; \gamma_1^1, \gamma_3^1, \gamma_1^2, \gamma_3^2 \geq 0.
    $$
On the basis of SAA with $5\cdot 10^7$ samples one observes that the worst-case risk is insensitive to dependence in the lower tail. The worst-case risk is attained for a $\gamma_3^2 =1$ (upper tail) with an AV@R at level 0.99 of $14.71$. This is illustrated in Figure~\ref{fig:3_3_colormap}. The worst-case risk in the considered model is lower than the sum of the marginal AV@Rs which is equal to $15.49$; this value corresponds to the comonotonicity of all components. This confirms that the underlying assumption (i.e., that dependence in different regions can be modeled separately and that an expert's opinion limits the choices of copulas) reduces model uncertainty.

\begin{figure}[h!]
    \centering
    \includegraphics[width=0.9\textwidth]{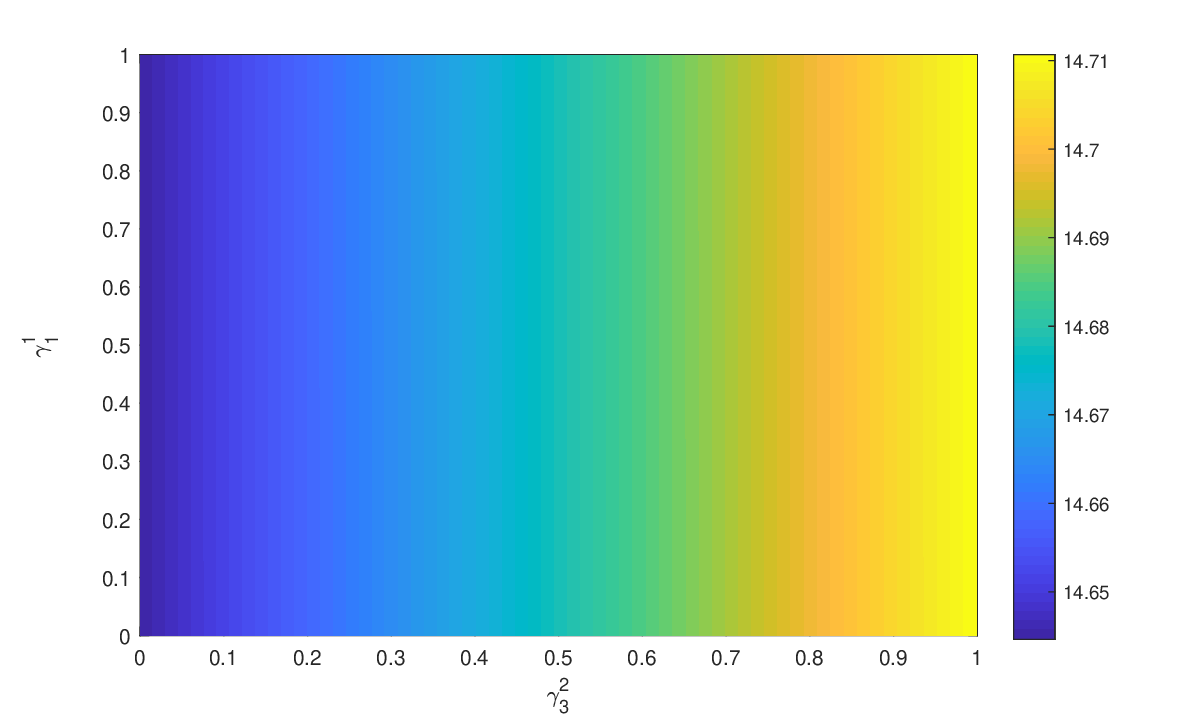}
    \caption{Color map of AV@R for parameters  $\gamma^1_1$ and $\gamma_3^2$.}\label{fig:3_3_colormap}
\end{figure}

In summary, in this case study SA is capable of quickly computing a reasonable estimate of the worst-case risk. Suitable worst-case copulas, encoded by the matrix $\bar{\bm{\gamma}}$, in a lower-dimensional mixture space can be determined by a combination of SA, the copula selection method described in Section~\ref{sec:switching}, and SAA.

\section{Applications}

Our method is flexible and can be used in multiple application domains. For the purpose of illustrating its applicability, we consider two case studies. The first example in Section~\ref{sec4:real} is based on a substantial amount of financial market data and allows a good calibration of copulas. Model risk can thereby be reduced.\footnote{In order to illustrate this claim, we include an additional case study in Appendix \ref{app_3}.} For cyber risk, the second example discussed in Section \ref{sec4:cyber}, only few observations are available which also increases the model risk.

\label{sec4:app}

\subsection{Financial markets}\label{sec4:real}

We apply our methodology to a data set spanning the time interval 2005/01/01 to 2019/12/31 that contains the daily closing values of the following stock indices: 
\begin{center}
\begin{tabular}{c | l }
i & Index\\ \hline 
1 &  S\&P 500 \\
2 & NASDAQ Composite  \\
3 &  Dow Jones Industrial Average \\
4 & DAX Performance Index  \\
5 & EURONEXT 100 \\
6 &  KOSPI Composite Index \\
7&  Nikkei 225
\end{tabular}
\end{center}
This data period also includes extreme events during the 2007 -- 2008 global financial crisis.
The 7-dimensional time series is labeled by trading days $t= 1, 2, \dots, 3358$ and quoted in US\$:

$$\p_t  
= 
\begin{pmatrix}
\p_{1,t}&
\p_{2,t}&
\p_{3,t }&
\p_{4,t}&
\p_{5,t}&
\p_{6,t}&
\p_{7,t}
\end{pmatrix}
$$
with $\p_{i,t}$ being the time $t$ US\$-price of index $i$, $i=1,2, \dots, 7$.  We consider a 7-dimensional random vector\footnote{For simplicity, we do not considered any dependence of the returns on time or on underlying factors. In practice, conditional distributions are typically important to reflect the market conditions.} ${\bf X} = (X_1, X_2, \dots, X_7)$ that models the negative returns of the terminal US\$-value of the indices over a 10-day horizon. The corresponding time series is given by
$${\bf x}_t  = \left( \{- \; 1 \} \; \cdot   \;
\frac
{\p_{i,t+10} - \p_{i,t}}
{\p_{i,t}}
\right)
_{i=1,2,\dots,7}, 
\quad\quad\quad 
t =1,2, \dots, \underbrace{3348}_{\quad =: \; D} $$
We investigate the robust AV@R at level $0.95$ over a 10-day horizon of a portfolio that invests an equal dollar amount into each index. To be more specific, we consider the robust AV@R of the losses $X = \sum_{i=1} ^7 X_i $.

\subsubsection{Marginal distributions}  

We apply a semi-parametric approach to the seven marginal distributions. For the central part of the distributions we linearly interpolate the empirical distribution. Less data are available in the tail, and we fit Generalized Pareto Distributions (GPD) to the data which allow a convenient extrapolation of samples. 

The CDF of a GPD with two parameters $\xi$ and $\vartheta$ is given by 
\begin{equation*}
    G_{\xi,\vartheta}(x) = \begin{cases} 
    1-\left( 1+\frac{\xi x}{\vartheta}\right)^{-1/\xi}, & \mbox{ if } \xi \neq 0 \\
    1 - \exp(-x/\vartheta), & \mbox{ if } \xi = 0.
    \end{cases}
\end{equation*}
The GPD is supported on $x \geq 0$, if $\xi \geq 0$, and on $0 \leq x \leq -\vartheta/\xi$, if $\xi < 0$. 

To be specific, for any $i =1,2, \dots, 7$, let $(x_{i,(t)})_{t=1,\cdots, D}$ be the ordered sample of  the data $(x_{i,t})_{t=1,\cdots, D}$ with $x_{i,(1)} \leq \cdots \leq x_{i,(D)}$.
As GPD approximates a tail distribution for the excesses above some high threshold, we let $ x_{i,l}$ and $x_{i,u}$ be suitably chosen thresholds of lower and upper parts. We apply a graphical diagnostic for the threshold choice; it is based on a mean excess plot that should be linear in the threshold for a GPD . For alternative, more sophisticated methods we refer to \ci{SM12}.
The two parameters $(\xi_{i,l},\vartheta_{i,l})$ and $(\xi_{i,u},\vartheta_{i,u})$ are then determined by  maximum likelihood estimation based on the lower and upper excess data $( x_{i,l} - x_{i,(1)}, \cdots, x_{i,l} - x_{i,(t_{i,l}-1)} )$ and $( x_{i,( t_{i,u}+1)}-x_{i,u} , \cdots, x_{i,(D)}-x_{i,u} )$, respectively. 
The estimated thresholds (i.e., the upper and lower boundaries: $x_{i,u} =  x_{i,( t_{i,u})}$, $x_{i,l} = x_{i,(t_{i , l})}$) and the corresponding parameters are reported in Table \ref{tab:GPDtailpara}.

\begin{table}[h]
\centering
\begin{tabular}{c|ccc|ccc} 
\multirow{2}{*}{Index $i$} & \multicolumn{3}{c|}{Lower tail} & \multicolumn{3}{c}{Upper tail} \\ \cline{2-7} 
                                    & boundary $x_{i,l}$   & shape  $\xi_{i,l}$ & scale $\vartheta_{i,l}$  & boundary $x_{i,u}$   & shape $\xi_{i,u}$  & scale $\vartheta_{i,u}$ \\ \hline
1                                   &  -0.0466&	0.3172&	0.0131&	0.0445&	0.2749&	0.0219   \\
2                                   & -0.0594&	0.1826&	0.0173&	0.0562&	0.2801&	0.0212 \\
3                                   &   -0.0355&	0.1791&	0.0124&	0.0446&	0.2355&	0.0209                                  \\
4 & -0.0578&	-0.1227&	0.0269&	0.0495&	0.1042&	0.0263  \\
5                                   &    -0.0640&	0.0453&	0.0176&	0.0656&	-0.0307&	0.0266 \\
6                                   &   -0.0658&	0.4326&	0.0211&	0.0600&	0.2211&	0.0335    \\
7                                   &    -0.0568&	0.3802&	0.0139&	0.0413&	0.1226&	0.0236   \\ \hline
\end{tabular}
\caption{The thresholds (boundaries) and the estimated shape and scale parameters of GPDs in the lower and upper tail parts.}\label{tab:GPDtailpara}
\end{table}

The linearly interpolated empirical distribution function truncated in $[x_{i,l}, x_{i,u}]$ for index $i$ is 
\begin{equation*}
    H_i(x) = \begin{cases}
            0 & x < x_{i,l}\\
            \frac{k-1}{t_{i,u}-t_{i,l}} + \frac{x - x_{i,(k)}}{(t_{i,u}-t_{i,l})\left(x_{i,(k+1)}-x_{i,(k)}\right)} & x_{i,(k)} \leq x < x_{i,(k+1)}\\
            1 & x \geq x_{i,u},
    \end{cases}
\end{equation*}
where $t_{i,l}, t_{i,u}$ are the indices in the ordered data corresponding to the lower and upper thresholds.

The distribution of $X_i$ is finally modeled by 
\begin{equation}
   F_i(x) = \begin{cases} 
        p_l  (1 - G_{\xi_{i,l},\vartheta_{i,l}}(x_{i,l}-x)) &  x \leq x_{i,l} \\
        p_l + (1-p_l-p_u)H_i(x) & x_{i,l}< x \leq x_{i,u} \\
        (1-p_u) + p_u G_{\xi_{i,u},\vartheta_{i,u}}(x-x_{i,u}) &  x > x_{i,u},
        \end{cases}
\end{equation}
where $p_l = P(x \leq x_{i,u})$ and $p_u = P(x \leq x_{i,u})$.

The results of Anderson-Darling and Cramér-von-Mises tests for the GPD lower and upper tails are reported in Table \ref{tab:GOF}. For the goodness of test of GPD, we follow the approach in \ci{CS01}. The results do not provide evidence against the estimated marginal distributions.

\begin{table}[]
\begin{tabular}{c|c|rrrrrrr}
\hline
Lower Tail                         & index     & \multicolumn{1}{c}{1} & \multicolumn{1}{c}{2}         & \multicolumn{1}{c}{3} & \multicolumn{1}{c}{4} & \multicolumn{1}{c}{5} & \multicolumn{1}{c}{6}         & \multicolumn{1}{c}{7}         \\ \hline
                                   & statistic & 0.3069	&0.3889	&0.4054	&0.2127&	0.4712	&0.2882&	0.6304                        \\
\multirow{-2}{*}{Anderson-Darling} & p value   & 0.6675	& 0.4972	&0.4642&	0.9255&	0.3742&	0.7022	&0.1497 \\ \hline
                                   & statistic & 0.0386	&0.0369	&0.0580	&0.0227	&0.0562&	0.0256&	0.0564                        \\
\multirow{-2}{*}{Cramér-von-Mises} & p value   & 0.6984	& 0.7462	&0.4371&	0.9628&	0.4876&	0.9024&	0.4177
\\ \hline

Upper tail                         & index     & \multicolumn{1}{c}{1} & \multicolumn{1}{c}{2}         & \multicolumn{1}{c}{3} & \multicolumn{1}{c}{4} & \multicolumn{1}{c}{5} & \multicolumn{1}{c}{6}         & \multicolumn{1}{c}{7}         \\ \hline
                                   & statistic & 0.5173	&0.4178&	0.4654&	0.5473&	0.3104&	0.4984&	0.2146                        \\
\multirow{-2}{*}{Anderson-Darling} & p value   & 0.2706	& 0.4240&	0.3491&	0.2637&	0.7126&	0.3030&	0.9053                     \\ \hline
                                   & statistic & 0.0634	& 0.0432&	0.0747&	0.0463&	0.0362	&0.0694&	0.0186                        \\
\multirow{-2}{*}{Cramér-von-Mises} & p value   & 0.3582	& 0.6264	& 0.2660&	0.6104&	0.7917&	0.3106&	0.9827       \\ \hline
\end{tabular}
\caption{The results of Anderson-Darling and Cramér-von-Mises tests for the lower and upper GPD approximations.}\label{tab:GOF}
\end{table}

\subsubsection{Dependence}

Since AV@R focuses on the upper tails, we consider the following  
distortion functions with parameter $\alpha_1 = \alpha_2 = 0.04$ and $\alpha_0 = 1-\alpha_1 -\alpha_2$: 
    $$ 
    D_0(x) = \left\{
    \begin{array}{ll}
    \frac{x}{\alpha_0} & \mbox{if }x \le \alpha_0 \\
    1 & x > \mbox{if }\alpha_0 \\
    \end{array}
    \right. \quad \quad \quad\quad\quad
    D_1(x) = \left\{
    \begin{array}{ll}
    0 & \mbox{if }x \le \alpha_0 \\
    \frac{x-\alpha_0}{\alpha_1} & \mbox{if }\alpha_0 < x  \leq \alpha_0 + \alpha_1 \\ 1& 
\mbox{if }x > \alpha_0 + \alpha_1 \\
    \end{array}
    \right.$$
    \begin{equation}\label{dist:ex}
    D_2(x) = \left\{
    \begin{array}{ll}
    0  & \mbox{if }x \le \alpha_0 + \alpha_1 \\
    \frac{x-\alpha_0-\alpha_1}{\alpha_2}  & \mbox{if }\alpha_0 + \alpha_1 < x  \leq 1. \\
    \end{array}
    \right. \end{equation}

We split the data on the basis of the aggregate loss function 
$\{ \sum_{i=1} ^7 x_{i,t} \}_{t=1, \cdots, D}$ into three parts: extreme upper tail (-4\%), upper tail (4\%-8\%), and the remaining center and lower tail (8\%-100\%). A scatter plot of $(X_1,X_2), \cdots, (X_1,X_7)$ in Figure \ref{fig:scatter_uuc} illustrates the procedure of partitioning the data.
\begin{figure}[h!]
     \centering
     \includegraphics[width=0.9\textwidth]{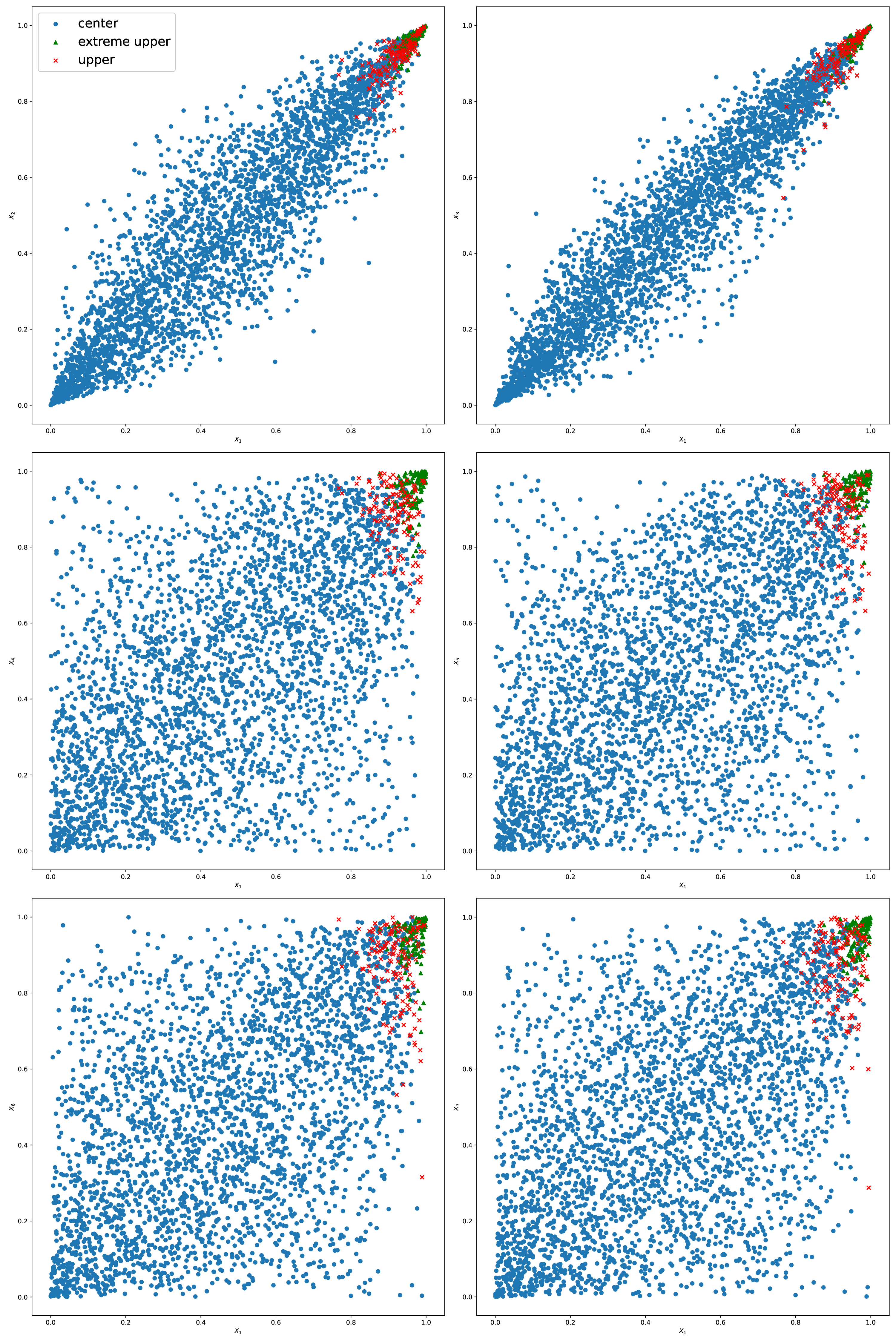}
     \caption{\label{fig:scatter_uuc} Scatter plot of $X_i$ against $X_1$ for $i=2, \cdots, 7$. The green triangles depict the data points in the extreme upper tail (-4\%), the red x the data points in the upper tail (4\%-8\%), and the blue circles the remaining data points (8\%-100\%).}  
\end{figure}

The dependence in the central part is modeled by a Gaussian copula corresponding to the estimated linear correlations. The estimates are based on the 92\% data (8\%-100\%) and can be found in Section \ref{app4:app}. For the tail parts, $K=16$ candidate copulas\footnote{The method is very flexible and could equally be applied to a larger set of copulas. The specific copulas are a potential choice due to an expert's opinion. This corresponds to model uncertainty that a priori is limited in this way.} are considered:
\begin{itemize}
\item Copulas $C_1$, $C_2$ are Gaussian, matching the linear correlation or Kendall's tau estimated from the upper (4\%-8\%) data, respectively.
\item The copulas $C_3$, $C_4$,$\dots$, $C_8$ are t-copulas whose parameters are calibrated on the basis of the upper (4\%-8\%) data:
\begin{itemize}
    \item the multivariate meta t-copulas $ (C_3, C_4, C_5) = ( C^{t}_{\nu^1,P^1}, C^{t}_{\nu^2,P^2}, C^{t}_{\nu^3,P^3} )$  with parameters $(\nu^l, P^l), \ l=1,2,3$; the superscript $l$ indicates the estimation method explained below;
    \item  the grouped t-copula $( C_6, C_7, C_8 ) =( C^{Gt}_{\bm \nu^1,\bm P^1}, C^{Gt}_{\bm \nu^2,\bm P^2}, C^{Gt}_{\bm \nu^3,\bm P^3} )$ that allow different subsets of the random variates to have different degrees of freedom parameters; we divide the indices into the three subgroups US, Europe, and Asia; for each $l=1,2,3,$ the vector $\bm \nu^l = (\nu_1^l, \nu_2^l, \nu_3^l)$ specifies the degrees of freedom for these three subgroups, and the superscript $l$ refers to the estimation method.
\end{itemize}
The first method ($l=1$) is ML estimation. The second method ($l=2$) exploits the approximated log-likelihood for the degrees of freedom parameter which increases the speed of the estimation.  
The third method ($l=3$) estimates the correlation matrix $P^3$ by Kendall’s tau for each pair and then estimates the scalar degree of freedom by ML given the fixed $P^3$. 
This method is useful when the dimension of the data is large, because the numerical optimization quickly becomes infeasible. 
The estimated correlation matrices as well as the degrees of freedom are not identical and sometimes even very different. 
\item The copulas $ C_9, C_{10}, \dots, C_{16}$ are constructed analogously to $ C_1, C_2, \dots, C_8$, but based on the extreme upper (-4\%) data.
\end{itemize}

The calibration results can be found in Section~\ref{app4:app}.

\subsubsection{Case studies}  

As a benchmark, we compute the AV@R  at level 0.95 when dependence is modeled by a single Gaussian copula estimated from the entire data set. The correlation matrix is given in Section~\ref{app4:app}, and the AV@R equals 0.514928 when the number of samples is $10^7$. 

We compare the benchmark to the algorithm based on DM copulas described in Section~\ref{sec3:sto_opt}. With an equal initial weight of $1/16$, a constant sample size $N_t = 10^6$ and step size $a=0.7$, the initial AV@R at level 0.95 corresponds to 0.652009 and is substantially higher than the benchmark. The DM method with copulas fitted to tail data provides a much better methodology in assessing downside risk than single Gaussian copulas. In fact, even if a DM method combines only Gaussian copulas for central and tail areas, the estimation results are often reasonable.\footnote{See Section \ref{app_3} for further evidence.} For the considered data, results are quite insensitive to the considered copulas, as long as they are fitted to different parts of the distribution and a DM copula is used. 

When running SA, the stopping time $t^*=t_{\min}$ equals 10 with an AV@R at level $0.95$ of $0.655033$ and an empirical standard deviation of $0.0019$ computed from the last ten iterations. The corresponding weights $\bar{\bm{\gamma}}_{10}^\top$ are 

{\tiny
        $$ 
        \begin{pmatrix}
        0.0497&0.0494&0.0676&0.0693&0.0503&0.0675&0.0693&0.0503&0.0630&0.0624&0.0696&0.0699&0.0616&0.0696&0.0694&0.0613 \\
        0.0595&0.0598&0.0643&0.0657&0.0618&0.0650&0.0624&0.0611&0.0624&0.0616&0.0649&0.0652&0.0628&0.0638&0.0573&0.0623 
        \end{pmatrix}
        $$}
 with an increment $ \bar{\bm{\gamma}}_{10} - \bar{\bm{\gamma}}_{9} $  with components of very small modulus (roughly less than 1/1000).
 
Now we switch to SAA. Setting $K^*=3$, our procedure selects on the basis of the estimated $\bar{\bm{\gamma}}_{10}$ for the application of SAA the copulas $C_4$ (t-copula estimated from the upper (4\%-8\%) data using approximate ML), $C_{12}$ (t-copula estimated from the extreme upper (-4\%) data using approximate ML), and $C_{14}$ (grouped t-copula estimated from the extreme upper (-4\%) data using ML). 
SAA with sample size $10^7$ and grid size $0.1$ applied to the corresponding three-dimensional grid picks only one copula, $C_{12}$ (t-copula estimated from the extreme upper (-4\%) data using approximate ML), associated with an AV@R of $0.655452$. The numerical analysis confirms that the AV@R values are insensitive to $\bar{\bm{\gamma}}$ near the solution. In the current example, the worst case AV@R is very close to the sum of the marginal AV@Rs, i.e., to comonotonic dependence of all components, that equals $0.656120$.

In summary, when computing AV@R at level 0.95, in the current example a reasonable amount of tail data is available to estimate the dependence of the factors in different parts of the distribution. In contrast to a single Gaussian copula, the DM method provides solutions in the considered family that are not very sensitive to the choice of the estimated component copulas. But instead of making ad hoc assumptions that select a specific components copula a priori, our algorithm demonstrates explicitly the strength of the DM method, identifies and substantiates the insensitivity to components a posteriori and finally reduces the dimensionality of the problem in the worst-case analysis.

\subsubsection{Robustness}

Another question one may ask is how the initial choice of $\alpha_0, \alpha_1, \alpha_2$ influences the final result of the algorithm. In fact, since the distortion functions in eq. \eqref{dist:ex} and thus the regions of different copulas depend on these parameters, we cannot expect that the copulas are invariant if $\alpha_0, \alpha_1, \alpha_2$ change.

We follow the same methodology as described above but vary $\alpha_1$. Recall that $\alpha_0 = 1-\alpha_1 - \alpha_2$ and that $\alpha_2$ is equal to $\alpha_1$. For a given $\alpha_1$, we split the data into three parts depending on the value of $\alpha_0, \alpha_1$ and $\alpha_2$; the data are segmented into the extreme upper $-\alpha_1\cdot 100\%$ part, the upper $\alpha_1 \cdot100\% - 2 \cdot \alpha_2 \cdot 100\%$, and the remaining part. We apply this procedure for $\alpha_1$: $0.03, 0.04, 0.05, 0.06$. 

Table \ref{tab:alpha} displays the worst case DM AV@R values together with the copulas chosen by our algorithm for the application of SAA when the dimension of the simplex is reduced. Both the  worst case DM AV@R and the corresponding copulas for both $D_1$ and $D_2$, in both cases $C_{12}$, are robust with respect to the choice of $\alpha_1$. However, as indicated in the third row of Table \ref{tab:alpha}, when we reduced the dimension before applying SAA and select $K^*$  candidate copulas from the SA results, different candidates are picked.

\begin{table}[h]
\begin{tabular}{c|cccc}
\hline
$\alpha_1$                & 0.03         & 0.04         & 0.05        & 0.06        \\ \hline
Worst case DM AV@R        & 0.655472     & 0.655452     & 0.655967    & 0.655872    \\
\small Dimension reduction: selected copulas& C8, C12, C14 & C4, C12, C14 & C4, C8, C12 & C4, C8, C12 \\ \hline
\end{tabular}
\caption{The AV@R values of the worst case DM copula (corresponding to $C_{12}$ for both $D_1$ and $D_2$) and the selected candidate copulas for the SAA procedure with $K^*=3$.}
\label{tab:alpha}
\end{table}

\subsection{Cyber risk}\label{sec4:cyber}
In an application to cyber risk, we study cyber incidents in USA from Privacy Rights Clearinghouse (\url{https://privacyrights.org/}) collected from January 2005 until October 2019.
For a time window ending in 2016 the data set was also analyzed  by \ci{EJ18}. We consider loss records in periods of two months and rearrange the data accordingly.  This reduces the number of zero entries and admits a tractable analysis that does not separate zero entries from strictly positive losses.\footnote{In contrast to our simplified approach, \ci{EJ18} build their analysis on a methodology described  in \ci{ec12} that expresses the joint probability function by copulas with discrete and continuous margins. Our algorithmic approach can also be applied to their methodology. The statistical estimation is, however, more difficult in this case.} 

The data records contain information on the period, the number of events in each period and the corresponding losses. We consider five types of breaches:\footnote{The description was obtained from the website \url{https://privacyrights.org/}.}

\begin{center}
\begin{tabular}{c | l }
i & Data Breaches Type  (Number of zero data): description \\ \hline 
1 & DISC (0): Unintended Disclosure Not Involving Hacking, Intentional Breach or Physical Loss \\
2 & HACK (0): Hacked by an Outside Party or Infected by Malware \\
3 & INSD (17): Insider - employee, contractor or customer\\
4 & PHYS (5): Physical - paper documents that are lost, discarded or stolen \\
5 & PORT (16): Portable Device - lost, discarded or stolen laptop, smartphone, memory stick, etc.
\end{tabular}
\end{center}

The data set contains 89 two months periods. The number of dates with observations of zero losses or no incidents is provided in parenthesis. In order to simplify the analysis, we replace all zero entries by a uniform $(0,1)$ random variable; since the severity for non-zero losses is typically on the order of $10^3$ or more, this approach does not substantially modify the data, but admits a simplified model with strictly positive marginal densities.

\subsubsection{Losses due to data breaches}

We assume that the two months breach records can be modeled by a 5-dimensional random vector ${\bf L} = (L_1, L_2, \dots, L_5)$. We employ a loss distribution approach to the marginal distributions, i.e., 
$$ L_i = \sum_{j=1}^{N^i} R_j^i$$
where $R_j^i, j=1, \cdots, N^i$ are iid random variables representing the severity of individual loss records and $N^i$ signifies the random number of losses. The dependence among ${\bf L}$ is captured by a copula which will be modelled as a DM copula. The details of selection and calibration will be given below. In general, one is interest in measuring the risk of some functional of ${\bf L}$. As an illustrative example, we focus on the AV@R at level $0.95$  of the sum of its components, $X = \sum_{i=1}^5 L_i.$

\subsubsection{Marginal distributions}  
Motivated by \ci{EJ18}, we model the frequency and the severity of loss records separately, choosing a lognormal distribution for the severity and a negative binomial distribution for the number of losses in each period. We estimate the parameters of the distributions and summarize the results in Table \ref{tab:4_2_marginal}; for the negative binomial we applied MLE, for the lognormal unbiased estimates of mean and variance of the log-data.
\begin{table}[h]
\centering 
\begin{tabular}{c|cc|cc}
     & \multicolumn{2}{c|}{Negative binomial} & \multicolumn{2}{c}{Lognormal} \\ 
Type & $r$                & $p$               & $\mu$         & $\sigma$      \\ \hline
1    & 2.8684             & 0.1209           & 9.8543        & 2.4364        \\
2    & 1.6333             & 0.0543            & 11.7851       & 2.5086        \\
3    & 0.9250             & 0.1196            & 6.9622        & 4.4965        \\
4    & 1.3117             & 0.0632            & 7.9432        & 2.9350        \\ 
5    & 0.9685             & 0.0685            & 7.9445       & 4.7539       \\ \hline
\end{tabular}
\caption{Estimation results for the loss frequency and severity of $\bf L$.}
\label{tab:4_2_marginal}
\end{table}

\noindent For the implementation of our algorithm, we finally generate and store $10^7$ samples of each distribution.

\subsubsection{Dependence}

Since the AV@R focuses on the upper tails, we continue to use the distortion functions in \eqref{dist:ex}, again choosing $\alpha_1 = \alpha_2 = 0.04$ and $\alpha_0 = 1-\alpha_1 -\alpha_2$. If AV@R at level 0.95 is computed by a corresponding DM copula, the dependence on the central and lower part (captured by $D_0$) is very low; this is confirmed in numerical experiments. For this reason, we focus on a particularly simple approach and use a Gaussian copula for this part with linear correlation estimated from the data. The correlation matrix $\Sigma_1$ is given in Section \ref{app4:cyb} in the appendix.

For the upper tails (captured by $D_1$ and $D_2$), we consider $K=8$ candidate copulas, namely two Gaussian copulas, two t- copulas, two Gumbel copulas, and two vine copulas; for the latter we refer to \ci{DBCK13} for further information. More specifically, the copulas are estimated as follows; the corresponding parameters for Gaussian, t and vine copulas are given in Section \ref{app4:cyb}:
\begin{itemize}
    \item $C_1$: Gaussian copula with the estimated linear correlation $\Sigma_1$;    
    \item $C_2$: Gaussian copula that matches the estimated Kendall's tau with corresponding correlation matrix $\Sigma_2$;
    \item $C_3$: t-copula with parameters $\nu_1$ and $P_1$ estimated by MLE;    
    \item $C_4$: t-copula with parameters $\nu_2$ and $P_2$ with $P_2$ matching Kendall's tau  and  $\nu_2$ estimated by MLE;
    \item $C_5$: Gumbel copula estimated by MLE with parameter $ \theta = 1.875123$;
    \item $C_6$: Gumbel copula estimated on the basis of a minimal Cramér-von Mises distance according to \ci{HMM13} with parameter $ \theta = 1.000061$;
    \item $C_7$: Regular vine copula estimated according to AIC;
    \item $C_{8}$: Regular vine copula estimated according to BIC.
\end{itemize}

\subsubsection{Case studies}  

As a benchmark, we compute the AV@R  at level 0.95 when dependence is modeled by the single Gaussian copula $C_1$ estimated from the entire data set. The unit for the reported AV@R values is always one million. The estimated AV@R at level 0.95  equals $ 45.6533$ in this case on the basis of $10^7$ samples. If we use copulas $C_2$, $C_3$, ... , $C_8$ points estimates range from about 45.2 to 53.1 with significant sampling error. As in Section~\ref{sec4:real}, we compare this benchmark to the result of the algorithm with DM copulas that was described in Section~\ref{sec3:sto_opt}. The DM approach provides a more sophisticated analysis of the worst case. 

With an equal initial weight of 1/8, a constant sample size $N_t = 10^6$ and step size $a = 0.7$, the initial AV@R at level 0.95 corresponds to $49.4159$.  Stopping SA according to the our stopping rule at $t^* = 10$, we obtain an estimated AV@R at level $0.95$ of $53.4793$ with an empirical standard deviation of $3.8764$ computed from the last ten iterations. The corresponding $\bar{\bm{\gamma}}^\top$ equals
        $$
        \begin{pmatrix}
        0.1187 & 0.1191 & 0.1232 & 0.1193 & 0.1307 & 0.1307 & 0.1291 & 0.1292 \\
        0 & 0 & 0 & 0 & 0.5751 & 0.4249 & 0 & 0
        \end{pmatrix}$$
 with increments $ (\bar{\bm{\gamma}}_{10} - \bar{\bm{\gamma}}_{9})^\top$ of an order of 1/500 or less.

Setting $K^*=3$ and switching to SAA, our algorithm selects  the copulas $C_5$ (Gumbel copula with $\theta = 1.875123$), $C_{6}$ (Gumbel copula with $\theta = 1.000061$), and $C_{8}$ (Regular vine copula according to BIC, Table \ref{tab:vine_bic}) on the basis of the estimate $\bar{\bm{\gamma}}_{10}$. Thus, SAA needs to be applied to a three-dimensional grid on
    $$\gamma_5^i + \gamma_6^i + \gamma_8^i =1, \ i=1,2,\; \quad \gamma^i_j \geq 0, \  i=1,2, j=5,6,8.
    $$
With a sample size of $10^7$ sample size and $0.1$ grid size, SAA selects the Gumbel copula $C_{5}$ for $D_1$ and the Gumbel copula $C_6$ for $D_2$ with a worst-case AV@R at level $95\%$ of $53.6990$. For the distortion $D_1$, the copula $C_6$ leads to almost the same result, i.e., for $D_1$ the sensitivity of the AV@R with respect to $C_{5}$ and $C_6$ is almost zero. This worst case AV@R may be compared to the comonotonic case, i.e., the sum of the marginal AV@Rs, which equals $55.133051$. This shows that within our setting (that limits the admissible dependence structure on the basis of an expert's opinion) model uncertainty is already reduced.

In summary, when computing AV@R at level 0.95, DM methods provide an excellent method for identifying the relevant low-dimensional dependence structures, when many data are available as illustrated in Section~\ref{sec4:real}. In the current example on cyber risk, data are scarce and tail copulas are chosen ad hoc on the basis of an expert's opinion. In this case, our algorithm easily identified the worst-case dependence and reduces the dimensionality at the same time. If only few data are available in the tail, the choice of tail copulas is restricted by only few constraints and the sensitivities of the AV@R within this class are more significant. In all cases, the worst-case AV@R on the basis of the DM copula provides a substantially better understanding of downside risk than single copulas fitted to the whole data. 

\section{Conclusion}\label{sec5:conc}

Uncertainty requires suitable techniques for risk assessment. In this paper, we combined stochastic approximation and stochastic average approximation to develop an efficient algorithm to compute the worst case average value at risk in the face of tail uncertainty. Dependence was modelled by the distorted mix method that flexibly assigns different copulas to different regions of multivariate distributions.
The method is computationally efficient and allows at the same time to identify copulas in a lower-dimensional mixture space that capture the worst case with high precision. We illustrated the application of our approach in the context of financial markets and cyber risk. Distorted mix copulas can flexibly adjust the dependence structure in different regions of a multivariate distribution. Our research indicated that they provide a powerful and flexible tool for capturing dependence in both the central area and tails of distributions.


\bibliographystyle{jmr}
\bibliography{RCM}


\newpage 
\appendix

\section{Appendix}\label{sec:appendix}

\subsection{Proof of Lemma \ref{lem:zeta}}\label{proof:lem:zeta}

Define $\zeta(u,\bar{\bm{\gamma}}) = u + \frac{\alpha_0}{1-p} \mathsf{E}[\Psi^0- u]^+ + \sum_{i=1}^m \sum_{j=1}^K \frac{\alpha_i \gamma_j^i}{1-p} \mathsf{E}[\Psi^{ij}- u]^+$, where $\Psi^0$ and $\Psi^{ij}$ are random variables having the distribution $G_0$ and $G_{ij}$ in (\ref{eq:mixturedist}), respectively.
The finiteness of the function $\zeta$ is guaranteed by the existence of the AV@R, or equivalently by $\mathsf{E}|\Psi^0|<\infty$ and $\mathsf{E}|\Psi^{ij}|<\infty$ for each $i$ and $j$.
Moreover, a convex function $\zeta(\cdot,\bar{\bm{\gamma}})$ has finite right and left derivatives for any $\bar{\bm{\gamma}}$.
Observe that
    \begin{eqnarray*}
        \frac{\zeta(u',\bar{\bm{\gamma}})-\zeta(u,\bar{\bm{\gamma}})}{u'-u} = 1+\frac{\alpha_0}{1-p}\frac{\mathsf{E}[\Psi^0- u']^+ - \mathsf{E}[\Psi^0- u]^+}{u'-u} + \sum_{i=1}^m \sum_{j=1}^K \frac{\alpha_i \gamma_j^i}{1-p} \frac{\mathsf{E}[\Psi^{ij}- u']^+ - \mathsf{E}[\Psi^{ij}- u]^+}{u'-u}.
    \end{eqnarray*}
When $u'>u$,
    \begin{equation*}
        \frac{\mathsf{E}[\Psi^0- u']^+ - \mathsf{E}[\Psi^0- u]^+}{u'-u} = \begin{cases}
            -1 \mbox{ if } \Psi^0 \geq u' \\
            0 \mbox{ if } \Psi^0 \leq u \\
            \mathsf{E}\left[\frac{-\Psi^0+u}{u'-u}\right] \in (-1,0) \mbox{ if } u<\Psi^0<u'.
        \end{cases}
    \end{equation*}
Then there exist $\rho(u,u')\in[0,1]$ for which
    \begin{equation*}
        \frac{\mathsf{E}[\Psi^0- u']^+ - \mathsf{E}[\Psi^0- u]^+}{u'-u} = -(1-\mathsf{P}(\Psi^0\leq u')) - \rho(u,u')(\mathsf{P}(\Psi^0\leq u')-\mathsf{P}(\Psi^0\leq u)).
    \end{equation*}
By letting $u'\downarrow u$, we have $\mathsf{P}(\Psi^0\leq u')$ converges to $\mathsf{P}(\Psi^0\leq u)$ which makes
    \begin{equation*}
        \lim_{u'\downarrow u} \frac{\mathsf{E}[\Psi^0- u']^+ - \mathsf{E}[\Psi^0- u]^+}{u'-u} = \mathsf{P}(\Psi^0\leq u)-1.
    \end{equation*}
Similarly, we can compute
    \begin{eqnarray*}
        \lim_{u'\downarrow u} \frac{\zeta(u',\bar{\bm{\gamma}})-\zeta(u,\bar{\bm{\gamma}})}{u'-u} &=& 1+\frac{\alpha_0}{1-p}(\mathsf{P}(\Psi^0\leq u)-1) + \sum_{i=1}^m \sum_{j=1}^K \frac{\alpha_i \gamma_j^i}{1-p} (\mathsf{P}(\Psi^{ij}\leq u)-1) \\
        &=& 1-\frac{1}{1-p}+\frac{\alpha_0}{1-p}\mathsf{P}(\Psi^0\leq u) + \sum_{i=1}^m \sum_{j=1}^K \frac{\alpha_i \gamma_j^i}{1-p} \mathsf{P}(\Psi^{ij}\leq u)
    \end{eqnarray*}
which is $\frac{\partial^+\zeta}{\partial u}(u,\bar{\bm{\gamma}})$.
Analogously, we can compute $\frac{\partial^-\zeta}{\partial u}(u,\bar{\bm{\gamma}})$. The remaining results are now straightforward.

\subsection{Calibrations with large amounts of data}\label{app_3}

We provide an illustrative example supporting the claim that the DM method provides a good statistical framework for estimating the risk measure AV@R if a large amount of data is available. This claim refers, of course, to the DM method of \ci{LYY14}  itself. Computing a worst-case is thus not an issue in this case study. 
 
We consider a setting that modifies Example \ref{ex1} ($d=2$) as follows. Data are generated by a collection of models with $X = X_1 + X_2$ where the dependence between $X_1$ and $X_2$ is correctly described by one of the following copulas  
\begin{itemize}
    \item a Gaussian copula with $\rho = 0.7$;
    \item a t-copula with $\rho = 0.7, \nu = 1$;
    \item a Gumbel copula with $\theta = 1.7095$;
    \item a Cuadras-Augé copula with $\theta = 0.8$.
\end{itemize}
The Cuadras-Augé copula
$$C^{CA}(u_1, u_2) = \min(u_1, u_2)^\theta (u_1u_2)^{1-\theta}, \ \theta \in (0,1],$$ 
is upper tail dependent and an extreme value copula. The marginal distributions of $X_1$ and $X_2$ are inverse Gaussian with parameters $(\mu, \lambda)$ equal to $(1,0.5)$ and $(1, 1.2)$, respectively. For each model we compute the AV@R from SAA using the $5 \cdot 10^7$ number of samples as an approximation of the true AV@R; the results are displayed on the second row of Table \ref{tab:Exp_true} and denoted by `True AV@R'.

The numerical experiment is then conducted as follows. We generate data with sample size $10^3$ and $10^4$ from the given model (the given true copula). These are then used to estimate parameters of the following copulas and to finally compute AV@R measurements from them using a sample size $10^6$:
\begin{enumerate}
    \item a Gaussian DM copula with $\alpha_1=0.1$, whose component copulas for the three parts, $D_0, D_1, D_2$ are all Gaussian copulas; 
    \item  a Gaussian DM copula for the optimal $\alpha_1$ (as explained below);
    \item a Gaussian copula;
    \item a t-copula;
    \item a Gumbel copula. 
\end{enumerate}
Letting $\alpha_0 = 1-\alpha_1$, and $\alpha_1 = \alpha_2$, the distortion functions are:
$$ 
    D_0(x) = \left\{
    \begin{array}{ll}
    \frac{x}{\alpha_0} & \mbox{if }x \le \alpha_0 \\
    1 & x > \mbox{if }\alpha_0 \\
    \end{array}
    \right. \quad \quad \quad\quad\quad
    D_1(x) = \left\{
    \begin{array}{ll}
    0 & \mbox{if }x \le \alpha_0 \\
    \frac{x-\alpha_0}{\alpha_1} & \mbox{if }\alpha_0 < x  \leq \alpha_0 + \alpha_1 \\ 1& 
\mbox{if }x > \alpha_0 + \alpha_1 \\
    \end{array}
    \right.$$
    \begin{equation*}
    D_2(x) = \left\{
    \begin{array}{ll}
    0  & \mbox{if }x \le \alpha_0 + \alpha_1 \\
    \frac{x-\alpha_0-\alpha_1}{\alpha_2}  & \mbox{if }\alpha_0 + \alpha_1 < x  \leq 1. \\
    \end{array}
    \right. \end{equation*}
The results of the parameter estimation are displayed in Table \ref{tab:Exp_truepara}.
\begin{table}[H]
\begin{tabular}{c|c|c|rrrr}
\hline
Data size                & Copulas           & Parameters             & \multicolumn{1}{c}{Gaussian} & \multicolumn{1}{c}{t} & \multicolumn{1}{c}{Gumbel} & \multicolumn{1}{c}{Cuadras-Augé} \\ \hline
\multirow{12}{*}{$10^3$} & Gaussian DM & $\alpha_1$               & 0.1                          & 0.1                   & 0.1                        & 0.1                              \\
                         &             & $\hat{\rho}$ for $D_0$ & 0.5407                       & 0.5833                & 0.3273                     & 0.6718                           \\
                         &             & $\hat{\rho}$ for $D_1$ & -0.6668                      & -0.5114               & -0.7146                    & -0.5774                          \\
                         &             & $\hat{\rho}$ for $D_2$ & 0.016                        & 0.1752                & 0.2868                     & 0.0187                           \\
                         & Gaussian DM & $\alpha_1$               & 0.2                          & 0.12                  & 0.15                       & 0.1                              \\
                         &             & $\hat{\rho}$ for $D_0$ & 0.4514                       & 0.5849                & 0.2602                     & 0.6718                           \\
                         &             & $\hat{\rho}$ for $D_1$ & -0.6254                      & -0.5371               & -0.6672                    & -0.5774                          \\
                         &             & $\hat{\rho}$ for $D_2$ & 0.1467                       & 0.1889                & 0.2683                     & 0.0187                           \\ \cline{2-7} 
                         & Gaussian    & $\hat{\rho}$           & 0.7032                       & 0.7002                & 0.5973                     & 0.761                            \\
                         & t           & $\hat{\nu}$            & 3338586                      & 1.0062                & 5.3304                     & 1                                \\
                         &             & $\hat{\rho}$           & 0.6991                       & 0.7102                & 0.577                      & 1                                \\
                         & Gumbel      & $\hat{\theta}$         & 1.8341                       & 2.1882                & 1.6542                     & 3.1677                           \\ \hline
\multirow{12}{*}{$10^4$} & Gaussian DM & $\alpha_1$               & 0.1                          & 0.1                   & 0.1                        & 0.1                              \\
                         &             & $\hat{\rho}$ for $D_0$ & 0.5652                       & 0.5936                & 0.3456                     & 0.1191                           \\
                         &             & $\hat{\rho}$ for $D_1$ & -0.7711                      & -0.5744               & -0.6837                    & -0.562                           \\
                         &             & $\hat{\rho}$ for $D_2$ & -0.1119                      & -0.1026               & 0.1121                     & 0.1368                           \\
                         & Gaussian DM & $\alpha_1$               & 0.15                         & 0.12                  & 0.13                       & 0.1                              \\
                         &             & $\hat{\rho}$ for $D_0$ & 0.5185                       & 0.6076                & 0.3006                     & 0.1191                           \\
                         &             & $\hat{\rho}$ for $D_1$ & -0.6983                      & -0.605                & -0.6758                    & -0.562                           \\
                         &             & $\hat{\rho}$ for $D_2$ & -0.0157                      & -0.1003               & 0.1191                     & 0.1368                           \\ \cline{2-7} 
                         & Gaussian    & $\hat{\rho}$           & 0.7028                       & 0.6281                & 0.5906                     & 0.742                            \\
                         & t           & $\hat{\nu}$            & 4669186                      & 1.0107                & 7.7096                     & 1                                \\
                         &             & $\hat{\rho}$           & 0.7064                       & 0.6973                & 0.5938                     & 1                                \\
                         & Gumbel      & $\hat{\theta}$         & 1.8476                       & 2.0485                & 1.6837                     & 2.9518                           \\ \hline
\end{tabular}
\caption{Estimated parameters.}
\label{tab:Exp_truepara}
\end{table}

\begin{table}[H]
\begin{tabular}{c|c|rrrr}
\hline
True copulas              &                                                                               & \multicolumn{1}{c}{Gaussian} & \multicolumn{1}{c}{t} & \multicolumn{1}{c}{Gumbel} & \multicolumn{1}{c}{Cuadras-Augé} \\ \hline 
Data size                 & True AV@R                                             & 8.8405                       & 9.0755                & 9.0728                     & 9.2205                           \\ \hline
                          & \begin{tabular}[c]{@{}c@{}}Gaussian DM AV@R\\ ($\alpha_1 = 0.1$)\end{tabular}   & 9.2404                       & 9.2879                & 9.3345                     & 9.2093                           \\
                          & increment to true AV@R                                                                          & -0.3999                      & -0.2125               & -0.2617                    & 0.0112                           \\
                          & \begin{tabular}[c]{@{}c@{}}Gaussian DM AV@R\\ (optimal $\alpha_1$)\end{tabular} & \begin{tabular}[r]{@{}c@{}}8.7948\\ (0.2)\end{tabular}                      & \begin{tabular}[r]{@{}c@{}}9.1645\\ (0.12)\end{tabular}                & \begin{tabular}[r]{@{}c@{}}9.0836 \\ (0.15)\end{tabular}                     & \begin{tabular}[r]{@{}c@{}}9.2093\\ (0.1)\end{tabular}                           \\
                          & increment to true AV@R                                                                          & 0.0457                       & -0.0890               & -0.0107                    & 0.0112                           \\ \cline{2-6} 
                          & Gaussian AV@R                                                                 & 8.8089                       & 8.5252                & 8.5674                     & 8.9012                           \\
                          & increment to true AV@R                                                                          & 0.0316                       & 0.5503                & 0.5054                     & 0.3193                           \\
                          & t AV@R                                                                        & 8.8132                       & 9.0212                & 8.6768                     & 9.6490                           \\
                          & increment to true AV@R                                                                          & 0.0273                       & 0.0543                & 0.3960                     & -0.4285                          \\
                          & Gumbel AV@R                                                                   & 9.1872                       & 9.3263                & 9.0231                     & 9.5304                           \\
\multirow{-11}{*}{$10^3$} & increment to true AV@R                                                                          & -0.3467                      & -0.2508               & 0.0498                     & -0.3098                          \\ \hline
                          & \begin{tabular}[c]{@{}c@{}}Gaussian DM AV@R\\ ($\alpha_1 = 0.1$)\end{tabular}   & 9.1480                       & 9.1498                & 9.2478                     & 9.2560                           \\
                          & increment to true AV@R                                                                          & -0.3075                      & -0.0744               & -0.1750                    & -0.0467                          \\
                          & \begin{tabular}[c]{@{}c@{}}Gaussian DM AV@R\\ (optimal $\alpha_1$)\end{tabular} & \begin{tabular}[r]{@{}c@{}}8.8558  \\ (0.15)\end{tabular}                     & \begin{tabular}[r]{@{}c@{}}9.0162  \\ (0.12)\end{tabular}              & \begin{tabular}[r]{@{}c@{}}9.0541 \\ (0.13)\end{tabular}                    & \begin{tabular}[r]{@{}c@{}}9.2560 \\ (0.1)\end{tabular}                           \\
                          & increment to true AV@R                                                                          & -0.0153                      & 0.0593                & 0.0187                     & -0.0467                          \\ \cline{2-6} 
                          & Gaussian AV@R                                                                 & 8.8574                       & 8.6748                & 8.5548                     & 8.9677                           \\
                          & increment to true AV@R                                                                          & -0.0169                      & 0.4007                & 0.5181                     & 0.2529                           \\
                          & t AV@R                                                                        & 8.8647                       & 9.0564                & 8.6662                     & 9.6537                           \\
                          & increment to true AV@R                                                                          & -0.0242                      & 0.0190                & 0.4067                     & -0.4331                          \\
                          & Gumbel AV@R                                                                   & 9.2021                       & 9.2689                & 9.0558                     & 9.5103                           \\
\multirow{-11}{*}{$10^4$} & increment to true AV@R                                                                          & -0.3616                      & -0.1935               & 0.0171                     & -0.2897                          \\ \hline
\end{tabular}
\caption{ AV@R values in estimated models and increments to true AV@R.}
\label{tab:Exp_true}
\end{table}

The AV@Rs calculated in various estimated models and the resulting increments to the true AV@R are displayed in Table \ref{tab:Exp_true}. The Gaussian model with optimal $\alpha_1$ is estimated using  the empirical maximal likelihood. The main observation is that, despite being based on Gaussian copulas only, the DM method performs quite well in estimating the true AV@R. In particular, the Gaussian DM AV@R with optimal $\alpha_1$ outperforms all other models. This is particularly striking in the case of the  Cuadras-Augé copula, an extreme value copula. All copulas including the Gumbel copula perform worse in this case.

\subsection{Data in Section \ref{sec4:real}}\label{app4:app}

\subsubsection*{Dependence in the central part}
Dependence in the central part is modeled as the Gaussian copula whose correlation matrix consists of the estimated linear correlations. The estimated correlation matrix based on the 92\% data is 
\begin{center}
    $ \Sigma = \begin{bmatrix}
    1 & 0.9170 & 0.9494 & 0.4656 & 0.4706 &	0.5211 & 0.5012 \\
    0.9170 & 1 & 0.8185 & 0.4599 & 0.4495 &	0.5072 & 0.4724 \\
    0.9494 & 0.8185 & 1 & 0.4440 & 0.4389 &	0.4799 & 0.4698 \\
    0.4656 & 0.4599 & 0.4440 & 1 & 0.9152 & 0.2288 & 0.3191 \\
    0.4706 & 0.4495	& 0.4389 & 0.9152 &	1 &	0.2295 & 0.3060 \\
    0.5211 & 0.5072	& 0.4799 & 0.2288 &	0.2295 & 1 & 0.5717 \\
    0.5012 & 0.4724	& 0.4698 & 0.3191 &	0.3060 & 0.5717 & 1
    \end{bmatrix}$.
\end{center}

\subsubsection*{Dependence in the upper tail part}

We consider $K=16$ candidate copulas in the tail parts. 
The calibration results are summarized in the following.

\begin{itemize}
    \item $C_1$: Gaussian copula matching  the estimated linear correlation in the upper (4\%-8\%) data
    \begin{center}
    $ \Sigma_1 = \begin{bmatrix}
    1  & 0.8316  & 0.8727  & -0.2373 & -0.1814 & -0.1005 & -0.1121 \\
    0.8316  & 1  & 0.6434  & -0.1834 & -0.1717 & -0.0088 & -0.0635 \\
    0.8727  & 0.6434  & 1  & -0.2906 & -0.2663 & -0.2072 & -0.1561 \\
    -0.2373 & -0.1834 & -0.2906 & 1  & 0.7351  & -0.2962 & -0.2112 \\
    -0.1814 & -0.1717 & -0.2663 & 0.7351  & 1  & -0.2406 & -0.2076 \\
    -0.1005 & -0.0088 & -0.2072 & -0.2962 & -0.2406 & 1  & 0.0762  \\
    -0.1121 & -0.0635 & -0.1561 & -0.2112 & -0.2076 & 0.0762  & 1
    \end{bmatrix}$.
\end{center}
    \item $C_2$: Gaussian copula matching the estimated Kendall's tau in the upper (4\%-8\%) data 
    \begin{center}
    $ \Sigma_2 = \begin{bmatrix}
    1  & 0.8011  & 0.8571  & -0.2492 & -0.2328 & -0.1505 & -0.0449 \\
    0.8011  & 1  & 0.5947  & -0.1970 & -0.2012 & -0.0481 & -0.0090 \\
    0.8571  & 0.5947  & 1  & -0.3024 & -0.3445 & -0.1766 & -0.0952 \\
    -0.2492 & -0.1970 & -0.3024 & 1  & 0.7583  & -0.3262 & -0.2632 \\
    -0.2328 & -0.2012 & -0.3445 & 0.7583  & 1  & -0.2438 & -0.2332 \\
    -0.1505 & -0.0481 & -0.1766 & -0.3262 & -0.2438 & 1  & 0.0580  \\
    -0.0449 & -0.0090 & -0.0952 & -0.2632 & -0.2332 & 0.0580  & 1
    \end{bmatrix}$.
    \end{center}    
     \item $C_3$: t-copula estimated from the upper (4\%-8\%) data using ML 
    $$ \nu^1 = 10.8802, \quad P^1 = \begin{bmatrix}
    1 & 0.9763 & 0.9808 & 0.7991 & 0.8093 & 0.7712 & 0.7786 \\
    0.9763 & 1 & 0.9461 & 0.7950 & 0.7996 & 0.7855 & 0.7806 \\
    0.9808 & 0.9461 & 1 & 0.7725 & 0.7811 & 0.7382 & 0.7615 \\
    0.7991 & 0.7950 & 0.7725 & 1 & 0.9529 & 0.6949 & 0.7192 \\
    0.8093 & 0.7996 & 0.7811 & 0.9529 & 1 & 0.7050 & 0.7172 \\
    0.7712 & 0.7855 & 0.7382 & 0.6949 & 0.7050 & 1 & 0.7332 \\
    0.7786 & 0.7806 & 0.7615 & 0.7192 & 0.7172 & 0.7332 & 1
    \end{bmatrix}.$$

    \item $C_4$:  t-copula estimated from the upper (4\%-8\%) data using approximate ML  
    $$ \nu^2 = 4.8409, \quad P^2 = \begin{bmatrix}
    1 & 0.9884 & 0.9910 & 0.8914 & 0.8922 & 0.8662 & 0.8741 \\
    0.9884 & 1 & 0.9745 & 0.8876 & 0.8870 & 0.8765 & 0.8761 \\
    0.9910 & 0.9745 & 1 & 0.8813 & 0.8800 & 0.8507 & 0.8679 \\
    0.8914 & 0.8876 & 0.8813 & 1 & 0.9762 & 0.8283 & 0.8439 \\
    0.8922 & 0.8870 & 0.8800 & 0.9762 & 1 & 0.8318 & 0.8405 \\
    0.8662 & 0.8765 & 0.8507 & 0.8283 & 0.8318 & 1 & 0.8459 \\
    0.8741 & 0.8761 & 0.8679 & 0.8439 & 0.8405 & 0.8459 & 1
    \end{bmatrix}.$$

    \item $C_5$:  t-copula estimated from the upper (4\%-8\%) data using Kendall's tau and ML 
    $$ \nu^3 = 1.1237, \quad P^3 = \begin{bmatrix}
    1  & 0.8011  & 0.8571  & -0.2492 & -0.2328 & -0.1505 & -0.0449 \\
    0.8011  & 1  & 0.5947  & -0.1970 & -0.2012 & -0.0481 & -0.0090 \\
    0.8571  & 0.5947  & 1  & -0.3024 & -0.3445 & -0.1766 & -0.0952 \\
    -0.2492 & -0.1970 & -0.3024 & 1  & 0.7583  & -0.3262 & -0.2632 \\
    -0.2328 & -0.2012 & -0.3445 & 0.7583  & 1  & -0.2438 & -0.2332 \\
    -0.1505 & -0.0481 & -0.1766 & -0.3262 & -0.2438 & 1  & 0.0580  \\
    -0.0449 & -0.0090 & -0.0952 & -0.2632 & -0.2332 & 0.0580  & 1 
    \end{bmatrix}.$$

    \item $C_{6}$: Grouped t-copula estimated from the upper (4\%-8\%) data using ML 
    $$ \bm{\nu}^{1} = (4.8284, 4.4677, 12.8199),$$ 
    $$ \bm P^1 = \begin{bmatrix}
    1 & 0.9734 & 0.9780 \\
    0.9734 & 1 & 0.9388 \\
    0.9780 & 0.9388 & 1
    \end{bmatrix}, 
    \begin{bmatrix}
    1 & 0.9472 \\
    0.9472 & 1	\\
    \end{bmatrix},
    \begin{bmatrix}
    1	&0.7655	\\
    0.7655	&1	\\
    \end{bmatrix}.$$

    \item $C_{7}$: Grouped t-copula estimated from the upper (4\%-8\%) data using approximate ML 
    $$ \bm{\nu}^{2} = (1.2569,	1.7534,	13040952.7492),$$
    $$ \bm P^2 = \begin{bmatrix}
    1	&0.9713	&0.9793	\\
    0.9713	&1	&0.9375 \\
    0.9793 &0.9375 & 1
    \end{bmatrix}, 
    \begin{bmatrix}
    1	&0.9682	\\
    0.9682	&1	\\
    \end{bmatrix},
    \begin{bmatrix}
    1	&0.8517	\\
    0.8517	&1	\\
    \end{bmatrix}.$$

    \item $C_{8}$:  Grouped t-copula estimated from the upper (4\%-8\%) data using Kendall's tau and ML 
    $$ \bm{\nu}^{3} = (0.7287,	0.9180,	1.9547),$$
    $$ \bm P^3 = \begin{bmatrix}
    1	& 0.8011 & 0.8571 \\
    0.8011	& 1	& 0.5947 \\
    0.8571 & 0.5947 & 1
    \end{bmatrix}, 
    \begin{bmatrix}
    1	& 0.7583	\\
    0.7583	&1	\\
    \end{bmatrix},
    \begin{bmatrix}
    1	& 0.0580	\\
    0.0580 &1	\\
    \end{bmatrix}.$$
    
    \item $C_9$: Gaussian copula matching the estimated linear correlation in the extreme upper (-4\%) data 
    \begin{center}
    $ \Sigma_3 = \begin{bmatrix}
    1 & 0.9197 & 0.9592 & 0.5265 & 0.5336 & 0.5598 & 0.5913 \\
    0.9197 & 1 & 0.8860 & 0.4640 & 0.5141 & 0.5420 & 0.5060 \\
    0.9592 & 0.8860 & 1 & 0.4624 & 0.4786 & 0.4820 & 0.5998 \\
    0.5265 & 0.4640 & 0.4624 & 1 & 0.8503 & 0.5107 & 0.4501 \\
    0.5336 & 0.5141 & 0.4786 & 0.8503 & 1 & 0.4144 & 0.3761 \\
    0.5598 & 0.5420 & 0.4820 & 0.5107 & 0.4144 & 1 & 0.5388 \\
    0.5913 & 0.5060 & 0.5998 & 0.4501 & 0.3761 & 0.5388 & 1
    \end{bmatrix}$.
\end{center}
    \item $C_{10}$: Gaussian copula matching the estimated the Kendall's tau in the extreme upper (-4\%) data 
    \begin{center}
    $ \Sigma_4 = \begin{bmatrix}
    1 & 0.9088 & 0.9498 & 0.5068 & 0.4949 & 0.5044 & 0.5333 \\
    0.9088 & 1 & 0.8722 & 0.4565 & 0.4878 & 0.5028 & 0.4495 \\
    0.9498 & 0.8722 & 1 & 0.4615 & 0.4213 & 0.3868 & 0.5318 \\
    0.5068 & 0.4565 & 0.4615 & 1 & 0.8604 & 0.5050 & 0.4718 \\
    0.4949 & 0.4878 & 0.4213 & 0.8604 & 1 & 0.4223 & 0.3349 \\
    0.5044 & 0.5028 & 0.3868 & 0.5050 & 0.4223 & 1 & 0.4835 \\
    0.5333 & 0.4495 & 0.5318 & 0.4718 & 0.3349 & 0.4835 & 1
    \end{bmatrix}$.
    \end{center}   
    
    \item $C_{11}$: t-copula estimated from the extreme upper (-4\%) data using ML 
    $$ \nu^1_e = 25.9712, \quad P^1_e = \begin{bmatrix}
    1 & 0.9824 & 0.9909 & 0.8853 & 0.8961 & 0.8843 & 0.8906 \\
    0.9824 & 1 & 0.9738 & 0.8691 & 0.8926 & 0.8790 & 0.8685 \\
    0.9909 & 0.9738 & 1 & 0.8688 & 0.8825 & 0.8620 & 0.8919 \\
    0.8853 & 0.8691 & 0.8688 & 1 & 0.9553 & 0.8532 & 0.8378 \\
    0.8961 & 0.8926 & 0.8825 & 0.9553 & 1 & 0.8324 & 0.8264 \\
    0.8843 & 0.8790 & 0.8620 & 0.8532 & 0.8324 & 1 & 0.8620 \\
    0.8906 & 0.8685 & 0.8919 & 0.8378 & 0.8264 & 0.8620 & 1
    \end{bmatrix}.$$

    \item $C_{12}$: t-copula estimated from the extreme upper (-4\%) data using approximate ML 
    $$ \nu^2_e = 2.8781, \quad P^2_e = \begin{bmatrix}
    1 & 0.9864 & 0.9950 & 0.9062 & 0.9175 & 0.9321 & 0.9239 \\
    0.9864 & 1 & 0.9804 & 0.8876 & 0.9103 & 0.9243 & 0.9044 \\
    0.9950 & 0.9804 & 1 & 0.8965 & 0.9092 & 0.9202 & 0.9280 \\
    0.9062 & 0.8876 & 0.8965 & 1 & 0.9692 & 0.8957 & 0.8760 \\
    0.9175 & 0.9103 & 0.9092 & 0.9692 & 1 & 0.8838 & 0.8749 \\
    0.9321 & 0.9243 & 0.9202 & 0.8957 & 0.8838 & 1 & 0.9134 \\
    0.9239 & 0.9044 & 0.9280 & 0.8760 & 0.8749 & 0.9134 & 1
    \end{bmatrix}.$$

    \item $C_{13}$: t-copula estimated from the extreme upper (-4\%) data using Kendall's tau and ML 
    $$ \nu^3_e = 2.2042, \quad P^3_e = \begin{bmatrix}
    1 & 0.9088 & 0.9498 & 0.5068 & 0.4949 & 0.5044 & 0.5333 \\
    0.9088 & 1& 0.8722 & 0.4565 & 0.4878 & 0.5028 & 0.4495 \\
    0.9498 & 0.8722 & 1 & 0.4615 & 0.4213 & 0.3868 & 0.5318 \\
    0.5068 & 0.4565 & 0.4615 & 1 & 0.8604 & 0.5050 & 0.4718 \\
    0.4949 & 0.4878 & 0.4213 & 0.8604 & 1 & 0.4223 & 0.3349 \\
    0.5044 & 0.5028 & 0.3868 & 0.5050 & 0.4223 & 1 & 0.4835 \\
    0.5333 & 0.4495 & 0.5318 & 0.4718 & 0.3349 & 0.4835 & 1
    \end{bmatrix}.$$

    \item $C_{14}$: Grouped t-copula estimated from the extreme upper (-4\%) data using ML 
    $$ \bm{\nu}^{1}_e = (6.7057,	109.1713,	709778.4720),$$ 
    $$ \bm P^1_e = \begin{bmatrix}
    1	&0.9781	&0.9890	\\
    0.9781	&1	&0.9668 \\
    0.9890 & 0.9668 & 1
    \end{bmatrix}, 
    \begin{bmatrix}
    1	&0.9580	\\
    0.9580	&1	\\
    \end{bmatrix},
    \begin{bmatrix}
    1	&0.8848	\\
    0.8848	&1	\\
    \end{bmatrix}.$$

    \item $C_{15}$: : Grouped t-copula estimated from the extreme upper (-4\%) data using approximate ML 
    $$ \bm{\nu}^{2}_e = (1.0001,	1.1370,	1.7893),$$
    $$ \bm P^2_e = \begin{bmatrix}
    1	&0.9510	&0.9792	\\
    0.9510	&1	&0.9234 \\
    0.9792 & 0.9234 & 1
    \end{bmatrix}, 
    \begin{bmatrix}
    1	&0.9301	\\
    0.9301	&1	\\
    \end{bmatrix},
    \begin{bmatrix}
    1	&0.9092	\\
    0.9092	&1	\\
    \end{bmatrix}.$$

    \item $C_{16}$: : Grouped t-copula estimated from the extreme upper (-4\%) data using Kendall's tau and ML
    $$ \bm{\nu}^{3}_e = (1.2977,	1.2977,	1.5911),$$
    $$ \bm P^3_e = \begin{bmatrix}
    1	&0.9088	&0.9498	\\
    0.9088	&1	&0.8722 \\
    0.9498 & 0.8722 & 1
    \end{bmatrix}, 
    \begin{bmatrix}
    1	&0.8604	\\
    0.8604	&1	\\
    \end{bmatrix},
    \begin{bmatrix}
    1	&0.4835	\\
    0.4835	&1	\\
    \end{bmatrix}.$$
\end{itemize}

\subsubsection*{Single Gaussian copula estimated from the entire data set -- correlation matrix}

\begin{center}
   $  \begin{bmatrix}
   1 &	0.9388	& 0.9625 &	0.5961	& 0.6031	& 0.6302	& 0.6172 \\
   0.9388	& 1	& 0.8662	& 0.5906	& 0.5865	& 0.6197	& 0.5946 \\
   0.9625	& 0.8662	& 1	& 0.5775	& 0.5771	& 0.5965	& 0.5922 \\
   0.5961	& 0.5906	& 0.5775	& 1.0000	& 0.9333	& 0.4003	& 0.4699 \\
   0.6031	& 0.5865	& 0.5771	& 0.9333	& 1	& 0.4035	& 0.4623 \\
   0.6302	& 0.6197	& 0.5965	& 0.4003	& 0.4035	& 1.0000	& 0.6579 \\
   0.6172	& 0.5946	& 0.5922	& 0.4699	& 0.4623	& 0.6579	& 1
   \end{bmatrix}$.
\end{center}

\subsection{Data in Section \ref{sec4:cyber}}\label{app4:cyb}

\subsubsection*{Dependence in the central part}
Dependence in the central part is modeled as the Gaussian copula whose correlation matrix consists of the estimated linear correlations as 
 \begin{center}
    $ \Sigma_1 = \begin{bmatrix}
    1	& -0.0086	& -0.0224	& 0.0260	& -0.3324\\
    -0.0086	&1	& 0.1179	& -0.0210	&-0.2222\\
    -0.0224	&  0.1179	&1 &-0.1795	& 0.2620 \\
    0.0260	& -0.0210	&-0.1795	&1	& - 0.1342  \\
    -0.3324 & -0.2222	&  0.2620	& - 0.1342	&1 
    \end{bmatrix}.$
 \end{center}

\subsubsection*{Dependence in the tail parts}

\begin{itemize}

\item  
    $ \Sigma_2 = \begin{bmatrix}
    1	& 0.0605	& 0.1459	& 0.0437	& -0.1554\\
    0.0605	&1	& -0.0233	& 0.0313	&-0.2593\\
    0.1459	& -0.0233	&1 &-0.1787	& 0.3949 \\
    0.0260	& 0.0313	&-0.1787	&1	& -0.1677  \\
    -0.1554 & -0.2593	&  0.3949	& -0.1677	&1 
    \end{bmatrix}$
\item    $ \nu_1 = 27.5747, \quad P_1 = \begin{bmatrix}
    1	&0.6895	& 0.2085	&0.5510	& 0.3492\\
    0.6895	&1	& 0.2257	&0.6665	& 0.4838\\
    0.2085	& 0.2257	&1	& -0.3175	& -0.4300\\
    0.5510	& 0.6665	& -0.3175	&1	&0.8533 \\
    0.3492	& 0.4838	& -0.4300	& 0.8533	& 1
    \end{bmatrix}$
 \item     $ \nu_2 = 3.6372, \quad P_2 = \Sigma_2$     
\item $C_7$: The regular vine copula is estimated according to AIC. For more information, we refer to  \ci{DBCK13}. The estimation was conducted by the vine copula package in R. \url{ https://cran.r-project.org/web/packages/VineCopula/VineCopula.pdf}. The selected trees, pair copulas and the estimated parameters are given in Table \ref{tab:vine_aic}.
 \item $C_{8}$:  The regular vine copula is estimated according to BIC. For more information, we refer to  \ci{DBCK13}. The estimation was conducted by the vine copula package in R. \url{ https://cran.r-project.org/web/packages/VineCopula/VineCopula.pdf}. The selected trees, pair copulas and the estimated parameters are provided in Table \ref{tab:vine_bic}.
\end{itemize}

\newpage

   \begin{table}[h]
    \centering
    \begin{tabular}{lllll}
    Tree & pair        & copula                  & \multicolumn{2}{l}{parameters} \\ \hline
    1    & 3,4         & Frank                   & 3.63           &               \\
         & 5,2         & Frank                       & 6.12           & 2             \\
         & 5,1         & Frank               & 4.39            & 0.06          \\
         & 5,3         & Tawn type 2 180 degrees & 3.57          & 0.46          \\ \hline
    2    & 5,4 ; 3     & Tawn  type 2           &  4.12            &  0.37          \\
         & 1,2 ; 5     & Tawn type 2 180 degrees           &    1.78            &    0.41           \\
         & 3,1 ; 5     & Survival BB8                       & 6          &  0.17            \\ \hline
    3    & 1,4 ; 5,3   &  Tawn  type 1           & 3.61          &    0.39            \\
         & 3,2 ; 1,5   & Joe                      & 1.11          &         \\ \hline
    4    & 2,4 ; 1,5,3 &  Tawn type 2 180 degrees           & 1.6             & 0.31        
    \end{tabular}
    \caption{The structure, pair copulas, and parameters of the regular vine copula $C_7$ estimated according to AIC}\label{tab:vine_aic}.
    \end{table}
    
    \vspace{1cm}
    
    \begin{table}[h]
    \centering
    \begin{tabular}{lllll}
    Tree & pair        & copula                  & \multicolumn{2}{l}{parameters} \\ \hline
    1    & 3,4         & Frank                   & 3.63           &               \\
         & 5,2         & Frank                       & 6.12           & 2             \\
         & 5,1         & Frank               & 4.39            & 0.06          \\
         & 5,3         & Tawn type 2 180 degrees & 3.57          & 0.46          \\ \hline
    2    & 5,4 ; 3     &  Tawn  type 2           & 4.12         & 0.37       \\
         & 1,2 ; 5     &  Survival Joe             & 1.55          &               \\
         & 3,1 ; 5     &  Independence   &          &         \\ \hline
    3    & 1,4 ; 5,3   &  Tawn  type 1  & 3.49        & 0.39      \\
         & 3,2 ; 1,5   & Independence 
            &         &               \\ \hline
    4    & 2,4 ; 1,5,3 & Clayton                 & 0.52         &              
    \end{tabular}
    \caption{The structure, pair copulas, and parameters for the regular vine copula $C_8$ according to BIC}\label{tab:vine_bic}.
    \end{table}

\end{document}